\documentclass[twocolumn]{autart}
\usepackage{graphicx}
\usepackage{mathrsfs}
\usepackage{mathptmx}
\usepackage{amsmath} 
\let\theoremstyle\relax
\usepackage{amsthm} 
\usepackage[authoryear]{natbib}
\setcitestyle{authoryear,open={},close={},comma}
\usepackage{appendix}
\usepackage{amsfonts} 
\usepackage{amscd}
\usepackage{amssymb} 
\usepackage{array}
\usepackage{subfigure}
\usepackage{textcomp}
\usepackage{bm}
\usepackage{xcolor}
\usepackage{indentfirst}
\usepackage[OT2,T1]{fontenc}

\usepackage{algorithm}
\usepackage{algorithmic}

\newtheorem*{problem}{Problem Statement}
\newtheorem{example}{Example}
\newtheorem{defi}{Definition}
\newtheorem{theorem}{Theorem}

\newtheorem{rema}{Remark}

\newtheorem{lemma}{Lemma}

\newtheorem*{testprin}{I/O Test Principles}

\date{}

\begin{document}
\begin{frontmatter}

\title{Data-based approaches to learning and control by similarity between heterogeneous systems} 

\thanks[footnoteinfo]{The material in this paper was not	presented at any conference.}

\vspace{-0.8cm}

\author[Affi1,Affi3]{Chenchao Wang}\ead{chenchaow1999@163.com},    
\author[Affi1,Affi2,Affi3]{Deyuan Meng}\ead{dymeng@buaa.edu.cn}         

\address[Affi1]{School of Automation Science and Electrical Engineering, Beihang University (BUAA), Beijing 100191, PR China}             
\address[Affi2]{State Key Laboratory of CNS/ATM, Beijing 100191, PR China}
\address[Affi3]{The Seventh Research Division, Beihang University (BUAA), Beijing 100191, PR China}


\begin{abstract}
This paper proposes basic definitions of similarity and similarity indexes between admissible behaviors of heterogeneous host and guest systems and further presents a similarity-based learning control framework by exploiting the offline sampled data. By exploring helpful geometric properties of the admissible behavior and decomposing it into the subspace and offset components, the similarity indexes between two admissible behaviors are defined as the principal angles between their corresponding subspace components. By reconstructing the admissible behaviors leveraging sampled data, an efficient strategy for calculating the similarity indexes is developed, based on which a similarity-based learning control framework is proposed. It is shown that, with the application of similarity-based learning control, the host system can directly accomplish the same control tasks by utilizing the successful experience provided by the guest system, without having to undergo the trial-and-error process. All results in this paper are supported by simulation examples.
\end{abstract}
\end{frontmatter}

\section{Introduction}
\label{sec:introduction}

Learning-based control, as one of the most promising fields within control community, has attracted significant attention and popularity. Learning-based control takes direct inspirations from human's learning process (\citep{ZJiang2020FTSC}). When individuals endeavor to acquire new skills, they repetitively engage in specific tasks and gather experience from past failures, ensuring their ability to better accomplish the same tasks in the future.
In a similar manner, dynamical systems can also recursively benefit from the past and correct the control errors, with the guarantee of the enhanced control performances (\citep{DABristow2006IEEECSM}). Such learning-based control mechanisms that learn from one's own past experience have been extensively investigated, and some well-established control frameworks have been presented (see, e.g., \citep{SArimoto1984JRS,YLi2017arXiv,ILandau2011book}). All of these aforementioned control frameworks, which either design controllers or adjust adaptive parameters based on past experience to rectify control errors, have found widespread and successful applications in real-world industrial systems (see, e.g., \citep{MBertolini2021ESA}).

Another characteristic of the human learning process entails its inherent strong interactivity, based on which a novice can efficiently acquire new skills through learning from the advanced experience of some skilled experts. Just as human inevitably require to learn from others to achieve predetermined goals, the collaborative learning of multiple dynamical systems to achieve some unified objective is an essential and extensively discussed topic (see, e.g., \citep{JPoveda2019IJACSP,SHe2018TII,HWang2016TIE}). A simple example may be the leader-follower formation problem in multi-agent systems (\citep{ADorri2018IEEEACCESS}), in which the followers obtain information and experience from the leader to realize the tracking of the predetermined trajectory. To ensure the achievement of the coorperative objective among the multiple systems, the idea of leveraging experience provided by other systems is extensively adopted. Nevertheless, the existing learning-based control strategies somewhat exhibit weaknesses in the following two aspects:
\begin{enumerate}
	\item[W1)] Existing learning-based control strategies simply collect information from neighbors based on specific communication topology, failing to quantitively assess that which system's information is more beneficial;
	\item[W2)] Existing learning-based control strategies directly employ the (weighted) relative information among systems to design control strategies, without fully exploiting the potential of the experience from other systems.
\end{enumerate}
Therefore, in the scenarios where the host system is equipped with multiple external experience provided by guest systems, it is meaningful and urgent to develop a novel learning-based control framework, i.e., similarity-based learning control, to guide the learning and control for the host system. The proposed similarity-based learning control strategy is required to not only quantitively characterize the value of the guest system' experience but also develop guidelines for its efficient utilization.

Owing to its learning mechanism, the learning-based control typically yields more accurate control performances and lower model dependency. Simultaneously, with the advancements in computer science and storage technology, employing the sampled data generated during the operation of systems for control objectives has become increasingly reliable and convenient. Designing learning-based control strategies with sampled data can further reduce the dependency on model knowledge. There have been several results that presented learning-based control strategies within data-driven frameworks (see, e.g., \citep{ZHou2013InformationScience,CDPersis2019TAC,BChu2023CDC}). However, as mentioned earlier, these results pay few attention to how to efficiently exploit the successful experience provided by guest systems.

Motivated by aforementioned discussions, this paper aims to propose a similarity-based learning control framework by leveraging the sampled data. With the application of the sampled input/output (I/O) data, the proposed similarity-based learning control is expected to overcome the shortcomings W1) and W2) of the exisiting learning-based control strategies. Specifically, we are interested in the scenarios where the guest system has completed the specific task through a repetitive trial-and-error process and provided its successful experiences to the host system. The similarity-based learning control framework can assess the value of the received experience and ensure that the host system can efficiently leveraging the experience to accomplish the same task without incurring the costly trial-and-error process. The mechanism of the similarity-based learning control framework is depicted in Fig. \ref{fig-mechanism}.
\begin{figure}[!ht]
	\centering
	\includegraphics[width=0.9\columnwidth]{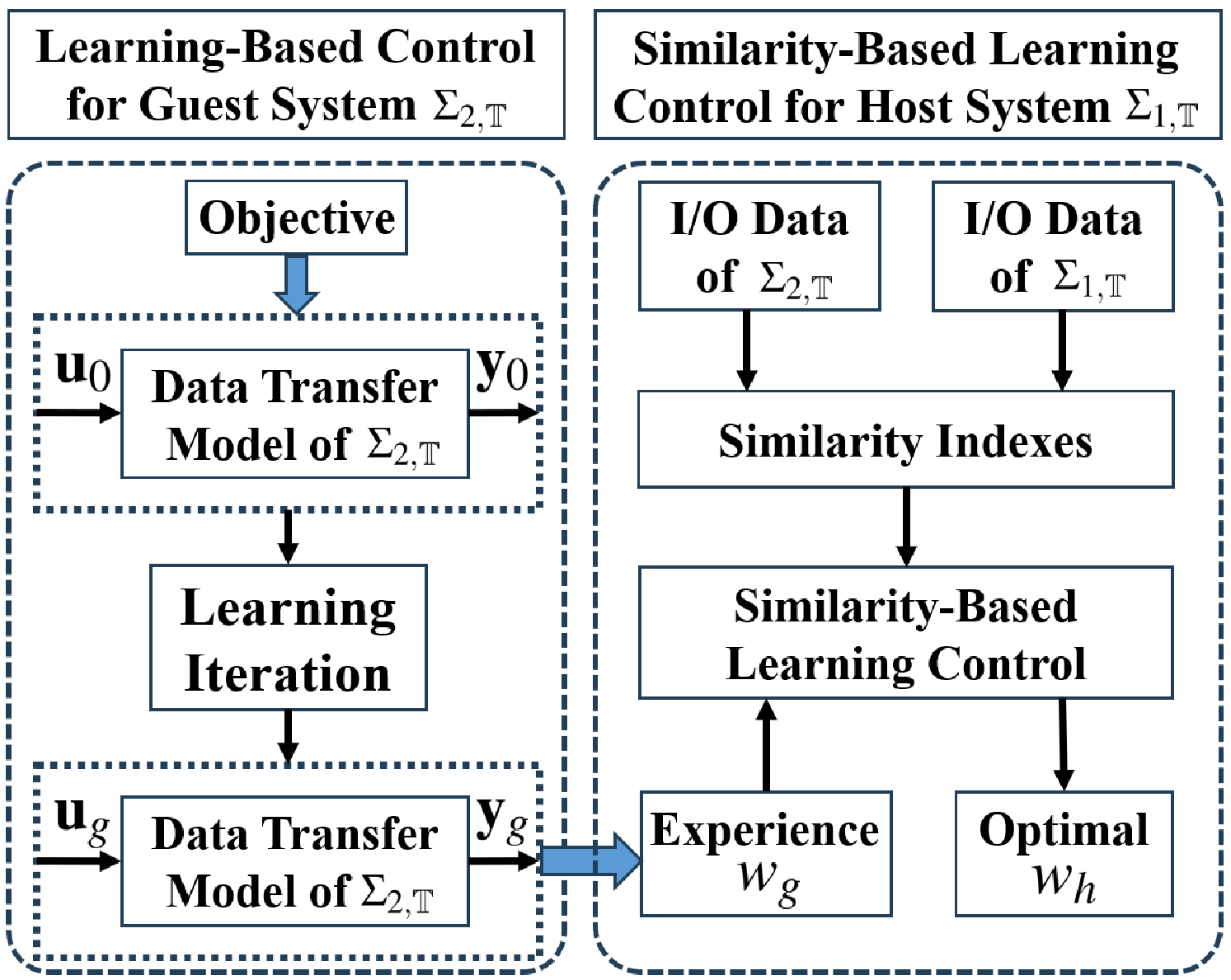}
	\caption{Similarity-based learning control exploiting sampled data.}
	\label{fig-mechanism}
\end{figure}
Main contributions of this paper can be summarized as follows.
\begin{enumerate}
	\item[C1)] We innovatively propose the basic definitions of similarity and similarity indexes between two admissible behaviors, which can qualitatively and quantitatively measure the benifits of guest system's successful experience of to the control of the host system, respectively;
	\item[C2)] By designing offline I/O test principles for heterogeneous linear time-varying (LTV) systems and exploiting the collected I/O data, we develop a data-based criterion for verifying the similarity and present an efficient strategy for calculating the similarity indexes;
	\item[C3)] By leveraging the calculated similarity indexes and exploiting helpful projection techniques, we establish a similarity-based learning control framework from the offline sampled data. As a result, this framework allows the host system to directly leverage the successful experience of the guest system to accomplish the specific tasks.
\end{enumerate}

The rest of this paper is organized as follows. We present the preliminaries on admissible behavior of LTV systems and formulate the similarity-based learning control problems in Section \ref{sec:Preliminaries and Problem Statement}. In Section \ref{subsec:Data-Based Construction For Admissible Behaviors}, we design the offline I/O test principles and reconstruct the admissible behaviors from the sampled data. Afterward, we introduce the definitions of similarity and similarity indexes, and develop a data-based criterion for verifying the similarity and a data-based strategy for efficiently calculating the similarity indexes in section \ref{subsec:Similarity and Similarity Indexes}. In Section \ref{sec:Similarity-Based Learning Control Strategy}, by exploiting the calculated similarity indexes and projection techniques, the similarity-based learning control framework is presented exploiting the sampled data. Finally, Section \ref{sec:Simulation Example} provides illstrative simulations, and Section \ref{sec:Conclusions} summarizes the contributions in this paper.

\emph{Notations:} Let $\mathbb{Z}_N=\{0,1,\cdots,N\}$ and $\mathbb{Z}_+=\{0,1,2,\cdots\}$. Let $\mathbb{R}$ be the set involving all real numbers, and $\mathbb{R}^n$ involves all $n$-dimensional real vectors whose entries locate in $\mathbb{R}$. For any matrix $A$, its transpose and kernel space are denoted as $A^{\rm T}$ and $\text{ker}\left(A\right)$, respectively. The linear space spanned by the columns of $A$ is denoted as $\text{span}(A)$. For arbitrary vectors $a,b\in\mathbb{R}^n$, the standard inner product $\langle a,b\rangle$ refers to $a^{\rm T}b$, and the induced norm is correspondingly defined as $\left\|a\right\|=\sqrt{\langle a,a \rangle}$. The identity and null matrices with appropriate dimensions are denoted as $I$ and $0$, respectively. Given $s_1,s_2,\cdots,s_n\in\mathbb{R}$, the symbol $\text{diag}(s_1,s_2,\cdots,s_n)$ represents the diagonal matrix whose diagonal entries are $s_1,s_2,\cdots,s_n$.

\section{Admissible behavior and problem statement} \label{sec:Preliminaries and Problem Statement}
The preliminaries of admissible behavior are firstly introduced. We consider two unknown heterogeneous LTV systems whose dynamics within the time duration $\mathbb{T}$ are represented as
\begin{equation} \label{statespacemodel}
	\Sigma_{i,\mathbb{T}}:\left\{
	\begin{aligned}
		x_i(t+1)&=A_i(t)x_i(t)+B_i(t)u_i(t) \\
		y_i(t)&=C_i(t)x_i(t)+D_i(t)u_i(t)
	\end{aligned}
	\right.,\ t\in\mathbb{T}, \ i\in\{1,2\}. \hfill
\end{equation}
Here, the subscripts $i=1$ and $i=2$ are employed to refer to the host system and guest system, respectively, and the host system $\Sigma_{1,\mathbb{T}}$ can receive experience from the guest system $\Sigma_{2,\mathbb{T}}$. Without loss of generality, the time duration is assumed to be $\mathbb{T}:=\mathbb{Z}_{T-1}$. The input and output are denoted as $u_i(t)\in\mathbb{R}^{n_u}$ and $y_i(t)\in\mathbb{R}^{n_y}$, respectively. The internal state with unknown dimension is denoted as $x_i(t)\in\mathbb{R}^{\bullet}$, and the unknown time-varying model matrices with appropriate dimension is represented by $\left\{A_i(t),B_i(t),C_i(t),D_i(t)\right\}$. It is worth mentioning that the results proposed in this paper can be implemented equally to the scenarios where $D_i(t)\equiv0$ for all $t\in\mathbb{T}$, and the introduction of $D_i(t)$ is solely for a generalized expression. In order to investigate the I/O relationship over the entire time duration $\mathbb{T}$, the following supervectors
\begin{equation} \label{eq-supervectors}
	\begin{aligned}
		\mathbf{u}_i&=\begin{bmatrix}
			u_i^{\rm T}(0),\ u_i^{\rm T}(1), \cdots, u_i^{\rm T}(T-1)
		\end{bmatrix}^{\rm T}, \\
		\mathbf{y}_i&=\begin{bmatrix}
			y_i^{\rm T}(0),\ y_i^{\rm T}(1), \cdots, y_i^{\rm T}(T-1)
		\end{bmatrix}^{\rm T}, \\
		\mathbf{x}_i&=\begin{bmatrix}
			x_i^{\rm T}(0),\ x_i^{\rm T}(1), \cdots, x_i^{\rm T}(T-1)
		\end{bmatrix}^{\rm T}
	\end{aligned}
\end{equation}
are introduced. For a vector $w_i=\text{col}\left(\mathbf{u}_i,\mathbf{y}_i\right)\in\mathbb{R}^{n_wT}$ where $n_w=n_u+n_y$, if there exists some (may be non-unique) state supervector $\mathbf{x}_i$ such that $\left(\mathbf{u}_i,\mathbf{y}_i,\mathbf{x}_i\right)$ satisfies (\ref{statespacemodel}), then $w_i$ is called as a \emph{$T$-length trajectory} of $\Sigma_{i,\mathbb{T}}$. To capture the I/O transfer characteristics, the \emph{behavior} of $\Sigma_{i,\mathbb{T}}$, denoted by $\mathcal{B}_{i}$, is defined as the set involving all $T$-length trajectories, i.e.,
\begin{equation} \nonumber
	\mathcal{B}_{i}=\left\{w_i\in\mathbb{R}^{n_wT}\left|\exists \mathbf{x}_i\text{ such that }\left(\mathbf{u}_i,\mathbf{y}_i,\mathbf{x}_i\right)\text{ satisfies }(\ref{statespacemodel})\right.\right\}.
\end{equation}
It is worth mentioning that the above definition only focuses on the I/O transfer characteristics, but neglects the initially stored energy in the system $\Sigma_{i,\mathbb{T}}$, which can be characterized by $x_i(0)$, and its influence on the system responses. Without loss of generality, throughout this paper, we assume that the unknown initial state of $\Sigma_{i,\mathbb{T}}$ is $x_i(0)=x_{i0}$. Due to this consideration, we introduce a category of \emph{$T$-length admissible trajectories}, denoted by $w_{i,x_{i0}}$, which refer to those $T$-length trajectories who start from $x_i(0)=x_{i0}$. Correspondingly, the \emph{admissible behavior} is defined as follows.
\vspace{10pt}
\begin{defi} \rm
	For the LTV system $\Sigma_{i,\mathbb{T}},\ i\in\{1,2\}$ under the initial state $x_i(0)=x_{i0}$, its admissible behavior, denoted by $\mathcal{B}_{i,x_{i0}}$, refers to the set involving all $T$-length admissible trajectories, i.e.,
	\begin{equation} \nonumber
		\mathcal{B}_{i,x_{i0}}=\left\{w_{i,x_{i0}}\in\mathbb{R}^{n_wT}\left|w_{i,x_{i0}}\in\mathcal{B}_i\text{ and }x_i(0)=x_{i0}\right.\right\}.
	\end{equation}
\end{defi}
\begin{rema} \rm
	From the above definition, the admissible behavior is essentially a subset of the behavior, i.e., $\mathcal{B}_{i,x_{i0}}\in\mathcal{B}_i$, since it is subject to extra constraints with respect to the initial energy. Admissible behavior also has specific engineering implications since practical systems always commence with some initially stored energy. A direct example may be the RLC circuits system where the the initial charge of the capacitor has an impact on the system responses (see, e.g., \citep{ROrtega2003TAC}).
\end{rema}
Based on the aforementioned preliminaries, we formulate the to-be-addressed problems in this paper as follows.
\vspace{10pt}
\begin{problem} \rm
	For the unknown host system $\Sigma_{1,\mathbb{T}}$ with initial state $x_1(0)=x_{10}$ and unknown guest system $\Sigma_{2,\mathbb{T}}$ with initial state $x_2(0)=x_{20}$, let their admissible behaviors be denoted as $\mathcal{B}_{1,x_{10}}$ and $\mathcal{B}_{2,x_{20}}$, respectively. This paper focuses on dealing with the following problems:
	\begin{enumerate}
		\item[P1)] Appropriate offline I/O test principles need to be designed, based on which the admissible behaviors $\mathcal{B}_{1,x_{10}}$ and $\mathcal{B}_{2,x_{20}}$ can be reconstructed by exploiting the sampled I/O data; \label{problem-1}
		\item[P2)] \label{problem-2}The definitions of similarity and similarity indexes between two admissible behaviors $\mathcal{B}_{1,x_{10}}$ and $\mathcal{B}_{2,x_{20}}$ need to be introduced. Moreover, a data-based criterion for verifying the similarity and a data-based strategy for calculating the similarity indexes need to be developed;
		\item[P3)] \label{problem-3} Suppose that the guest system $\Sigma_{2,\mathbb{T}}$ has accomplished its task through the trial-and-error process and achieved the desired trajectory $w_g\in\mathcal{B}_{2,x_{20}}$. A similarity-based learning control framework for the host system $\Sigma_{1,\mathbb{T}}$ needs to be proposed such that it can accomplish the same control task by exploiting the successful experience of $\Sigma_{2,\mathbb{T}}$ and the sampled data. As a result, we can find a solution $w_h\in\mathcal{B}_{1,x_{10}}$ such that the difference $\Vert w_g-w_h\Vert$ is minimized.
	\end{enumerate}
\end{problem}

\section{Data-based verification for similarity and similarity indexes} \label{sec:Data-Based Verification For Similarity and Similarity Indexes}
\subsection{Data-based reconstruction for admissible behaviors}
\label{subsec:Data-Based Construction For Admissible Behaviors}
In this subsection, we aim at recovering the admissible behaviors $\mathcal{B}_{1,x_{10}}$ and $\mathcal{B}_{2,x_{20}}$ by exploiting the sampled data, such that the problem P1) can be addressed. Compared to the system matrices $\left\{A_i(t),B_i(t),C_i(t),D_i(t)\right\}$, the admissible behavior $\mathcal{B}_{i,x_{i0}}$ can accurately capture the I/O transfer and initial state-output transfer characteristics of the system $\Sigma_{i,\mathbb{T}}$, without involving the non-unique internal states.
Owing to the absence of model knowledge, the admissible behaviors $\mathcal{B}_{1,x_{10}}$ and $\mathcal{B}_{2,x_{20}}$ need to be identified, and the accuracy of identification is closely related to the sufficiency of the sampled data.
 Therefore, it is necessary to design the appropriate offline I/O test principles to ensure the data sufficiency. The following strategy can guarantee the construction of an alternative data-based representation for the admissible behavior. For the systems $\Sigma_{i,\mathbb{T}}$ where $i\in\{1,2\}$, at least $n_uT+1$ times I/O tests need to be conducted. In each I/O test, the system $\Sigma_{i,\mathbb{T}}$ starts from an unknown but fixed initial state $x_{i0}$, and the $k$-th test input over the entire time duration $\mathbb{T}$ is denoted as $\mathbf{u}_i^k\in\mathbb{R}^{n_uT}$. Correspondingly, the $k$-th test output over $\mathbb{T}$ is denoted as $\mathbf{y}_i^k\in\mathbb{R}^{n_yT}$, and all test input/output data are collected as
\begin{equation} \nonumber
	\begin{aligned}
		U^{Test}_i&=\begin{bmatrix}
			\mathbf{u}^0_i,\ \mathbf{u}^1_i,\ \cdots,\ \mathbf{u}^{n_uT}_i
		\end{bmatrix}\in\mathbb{R}^{n_uT\times\left(n_uT+1\right)},\\
		Y^{Test}_i&=\begin{bmatrix}
			\mathbf{y}^0_i,\ \mathbf{y}^1_i,\ \cdots,\ \mathbf{u}^{n_uT}_i
		\end{bmatrix}\in\mathbb{R}^{n_yT\times\left(n_uT+1\right)}.
	\end{aligned}
\end{equation}
For each pair of test input and output, let $w_i^k=\text{col}\left(\mathbf{u}_i^k,\mathbf{y}_i^k\right)$ where $i\in\{1,2\}$ and  $k\in\mathbb{Z}_{n_uT}$.
It can be immediately observed that every vector $w_i^k$ is an admissible trajectory of $\Sigma_{i,\mathbb{T}}$.
To collect a sufficient number of representative admissible trajectories, the test inputs need to be designed accoring to the following test principles.
\vspace{10pt}
\begin{testprin} \rm
	For the LTV system $\Sigma_{i,\mathbb{T}},\ i\in\{1,2\}$, the following offline test principles need to be conducted:
	\begin{enumerate}
		\item In the initial I/O test, the test input is designed as
		\begin{equation} \label{eq-principle1}
			\mathbf{u}^0_i=\begin{bmatrix}
				0^{\rm T}_{n_u},\ 0^{\rm T}_{n_u},\ \cdots,\ 0^{\rm T}_{n_u}
			\end{bmatrix}^{\rm T}\in\mathbb{R}^{n_uT};
		\end{equation}
		\item In the later $n_uT$ I/O tests, the test inputs are designed to satisfy the following rank condition
		\begin{equation} \label{eq-principle2}
			\text{rank}\left(\begin{bmatrix}
				\mathbf{u}^1_i,\ \mathbf{u}^2_i,\ \cdots,\ \mathbf{u}^{n_uT}_i
			\end{bmatrix}\right)=n_uT.
		\end{equation}
	\end{enumerate}
\end{testprin}
By leveraging the sampled data collected building upon the above test principles, we can proceed to construct a data-based representation for the admissible behavior.
\vspace{10pt}
\begin{lemma}
	\rm \label{lemma-data-based representation}
	For the LTV system $\Sigma_{i,\mathbb{T}},\ i\in\{1,2\}$ with initial state $x_{i0}$, its admissible behavior $\mathcal{B}_{i,x_{i0}}$ constitutes an affine set. Moreover, let the test input $U^{Test}_i$ fulfill the test principles (\ref{eq-principle1}) and (\ref{eq-principle2}). Then a vector $\overline{w}_{i,x_{i0}}=\text{col}\left(\overline{\mathbf{u}}_i,\overline{\mathbf{y}}_i\right)\in \mathcal{B}_{i,x_{i0}}$ if and only if there exists some $g_i\in\mathbb{R}^{n_uT}$ such that
	\begin{equation} \label{eq-data-basedrepresentation}
		\begin{bmatrix}
			\overline{\mathbf{u}}_i\\
			\overline{\mathbf{y}}_i
		\end{bmatrix}=W_ig_i+w_{i}^0
	\end{equation}
	where
	\begin{equation} \nonumber
		W_i=\begin{bmatrix}
			w_i^1-w_i^0,w_i^2-w_i^0,\cdots,w_i^{n_uT}-w_i^0
		\end{bmatrix}
	\end{equation}
\end{lemma}
\begin{proof}
	The proof is divided into three steps.
	
	Step 1. We prove that the admissible behavior $\mathcal{B}_{i,x_{i0}}$ constitutes an affine set.
	In order to investigate the relationship among the input, initial state, and output of $\Sigma_{i,\mathbb{T}}$ over the entire time duration $\mathbb{T}$, we introduce the I/O transfer matrix $G_i$ and the initial state-output transfer matrix $L_i$, both of which can be steadily constructed by employing the model matrices $\left\{A_i(t),B_i(t),C_i(t),D_i(t)\right\}$ (see, e.g., \citep{DABristow2006IEEECSM}). Then by leveraging the supervectors in (\ref{eq-supervectors}), it is derived that
		\begin{equation}\nonumber
				\mathbf{y}_i=G_i\mathbf{u}_i+L_ix_i(0),\ i\in\{1,2\}.
			\end{equation}
		Under the specific initial state $x_i(0)=x_{i0}$, any admissible trajectory $w_{i,x_{i0}}\in\mathcal{B}_{i,x_{i0}}$ must be the solution of the non-homogeneous linear algebraic equation (LAE) described by
		\begin{equation} \nonumber
				\begin{bmatrix}
						-G_i,\ I
					\end{bmatrix}w_{i,x_{i0}}=L_ix_{i0}.
		\end{equation}
		Let $w_{i,x_{i0}}$ and $v_{i,x_{i0}}$ be two admissible trajectories of $\Sigma_{i,\mathbb{T}}$, and let $\alpha\in\mathbb{R}$ be arbitrary real number. Then the affine combination of $w_{i,x_{i0}}$ and $v_{i,x_{i0}}$ can be represented by $\alpha w_{i,x_{i0}}+\left(1-\alpha\right)v_{i,x_{i0}}$, which must fulfill
		\begin{equation} \nonumber
			\begin{bmatrix}
				-G_i,\ I
			\end{bmatrix}\left(\alpha w_{i,x_{i0}}+\left(1-\alpha\right)v_{i,x_{i0}}\right)=L_ix_{i0}.
		\end{equation}
		That is, the affine combination of arbitrary admissible tracjectories remains an admissible trajectory of $\Sigma_{i,\mathbb{T}}$. Therefore, the admissible behavior $\mathcal{B}_{i,x_{i0}}$ constitutes an affine set.
		
		Step 2. We illustrate that $\overline{w}_{i,x_{i0}}=\text{col}\left(\overline{\mathbf{u}}_i,\overline{\mathbf{y}}_i\right)$ is an admissible trajectory of $\Sigma_{i,\mathbb{T}}$ if (\ref{eq-data-basedrepresentation}) holds. From the designed offline I/O test principles (\ref{eq-principle1}) and (\ref{eq-principle2}), it can be concluded that every $w_i^k=\text{col}\left(\mathbf{u}^j_i,\mathbf{y}^j_i\right),\ i\in\{1,2\},\ k\in\mathbb{Z}_{n_uT}$
		is an admissible trajectory of $\Sigma_{i,\mathbb{T}}$. Let the vector $g_i$ be represented as
		\begin{equation} \label{eq-pf2lemma1_eq1}
			g_i=\begin{bmatrix}
				g_i^1,\ g_i^2,\ \cdots,\ g_i^{n_uT}
			\end{bmatrix}^{\rm T}.
		\end{equation}
		By substituting (\ref{eq-pf2lemma1_eq1}) into (\ref{eq-data-basedrepresentation}), the vector $\text{col}\left(\overline{\mathbf{u}}_i,\overline{\mathbf{y}}_i\right)$ can be equivalently expreesed as
		\begin{equation} \nonumber
			\begin{aligned}
				\begin{bmatrix}
					\overline{\mathbf{u}}_i\\
					\overline{\mathbf{y}}_i
				\end{bmatrix}&=\sum_{k=1}^{n_uT}g_i^k\left(w_i^k-w_i^0\right)+w_i^0\\
				&=\sum_{k=1}^{n_uT} g_i^kw_i^k+\left(1-\sum_{k=1}^{n_uT}g_i^k\right)w_i^0.
			\end{aligned}
		\end{equation}
		From the above equation, it is observed that $\text{col}\left(\overline{\mathbf{u}}_i,\overline{\mathbf{y}}_i\right)$ is essentially the affine combination of the admissible trajectories $\{w_i^0,w_i^1,\cdots,w_i^{n_uT}\}$. Since it has been proved that $\mathcal{B}_{i,x_{i0}}$ constitutes an affine set, we can conclude that $\text{col}\left(\overline{\mathbf{u}}_i,\overline{\mathbf{y}}_i\right)$ is an admissible trajectory.
		
		Step 3. We prove that any $\overline{w}_{i,x_{i0}}=\text{col}\left(\overline{\mathbf{u}}_i,\overline{\mathbf{y}}_i\right)\in\mathcal{B}_{i,x_{i0}}$ can always be expressed in the form of (\ref{eq-data-basedrepresentation}). From the test principle (\ref{eq-principle2}), the test inputs $\{\mathbf{u}_i^1,\mathbf{u}_i^2,\cdots,\mathbf{u}_i^{n_uT}\}$ form a set of bases of the linear space $\mathbb{R}^{n_uT}$. Therefore, for any input $\overline{\mathbf{u}}_i\in\mathbb{R}^{n_uT}$, it can be represented by
		\begin{equation} \label{eq-pf2lemma1_eq2}
			\overline{\mathbf{u}}_i=\sum_{k=1}^{n_uT}g_i^k\mathbf{u}_i^k
		\end{equation}
		where $g_i^k,\ k\in\mathbb{Z}_{n_uT}\backslash\{0\}$ represent the combination coefficients. By imposing the input $\overline{\mathbf{u}}_i$ to the system $\Sigma_{i,\mathbb{T}}$ under the initial state $x_{i0}$, corresponding outputs can be expressed as
		\begin{equation} \label{eq-pf2lemma1_eq3}
			\begin{aligned}
				\overline{\mathbf{y}}_i&=G_i\overline{\mathbf{u}}_i+L_ix_{i0}\\
				&=\sum_{k=1}^{n_uT}g_i^k\left(\mathbf{y}_i^k-L_ix_{i0}\right)+L_ix_{i0}.
			\end{aligned}
		\end{equation}
		Since the input components emcompasses all free variables in the admissible trajectory, based on (\ref{eq-pf2lemma1_eq2}) and (\ref{eq-pf2lemma1_eq3}), any admissible trajectory $\overline{w}_{i,x_{i0}}=\text{col}\left(\overline{\mathbf{u}}_i,\overline{\mathbf{y}}_i\right)$ can be represented by
		\begin{equation} \nonumber
			\begin{bmatrix}
				\overline{\mathbf{u}}_i\\
				\overline{\mathbf{y}}_i
			\end{bmatrix}=\sum_{k=1}^{n_uT}g_i^k\left(
			\begin{bmatrix}
				\mathbf{u}_i^k\\
				\mathbf{y}_i^k
			\end{bmatrix}-\begin{bmatrix}
				0_{n_uT}\\
				L_ix_{i0}
			\end{bmatrix}
			\right)+\begin{bmatrix}
			0_{n_uT}\\
			L_ix_{i0}
			\end{bmatrix}.
		\end{equation}
		According to the I/O test principle (\ref{eq-principle1}), the vector $w_{i}^0$ can be equivalently expressed as $w_{i}^0=\text{col}\left(0_{n_uT},L_ix_{i0}\right)$. Let $g_i=\begin{bmatrix}
			g_i^1,\ g_i^2,\ \cdots,\ g_i^{n_uT}
		\end{bmatrix}^{\rm T}$. Then any $\overline{w}_{i,x_{i0}}=\text{col}\left(\overline{\mathbf{u}}_i,\overline{\mathbf{y}}_i\right)\in\mathcal{B}_{i,x_{i0}}$ can be represented in the form of
		\begin{equation}\nonumber
			\begin{aligned}
				\begin{bmatrix}
					\overline{\mathbf{u}}_i\\
					\overline{\mathbf{y}}_i
				\end{bmatrix}&=\sum_{k=1}^{n_uT}g_i^k\left(w_i^k-w_i^0\right)+w_i^0\\
				&=W_ig_i+w_i^0.
			\end{aligned}
		\end{equation}
		Therefore, any admissible trajectory $\overline{w}_{i,x_{i0}}\in\mathcal{B}_{i,x_{i0}}$ can be expressed as (\ref{eq-data-basedrepresentation}). 
		
		Based on above three steps, the proof is completed.
\end{proof}
\begin{rema} \rm \label{rema-decomposition}
	Following the established data-based representation (\ref{eq-data-basedrepresentation}) for the admissible behavior, some helpful geometric properties of the admissible behavior can be further explored. By leveraging the offline I/O data, the admissible behavior $\mathcal{B}_{i,x_{i0}}$ can be decomposed into the sum of subspace and offset components (see, e.g., \citep{SBoyd2004convex}) as
	\begin{equation} \label{eq-databaseddecomposition}
		\begin{aligned}
			\mathcal{B}_{i,x_{i0}}&=\mathcal{W}_i+w_{i}^0,\\
			\mathcal{W}_i&=\text{span}\left(W_i\right).
		\end{aligned}
	\end{equation}
	From the offline test principles (\ref{eq-principle1}) and (\ref{eq-principle2}), $\mathcal{W}_i$ is a subspace in Euclidean space $\mathbb{R}^{n_wT}$ and is of dimension $n_uT$, which is exactly the number of free input channels.
	For simplicity, we define
	\begin{equation}\label{eq-unitbasesofW}
		H_i=\begin{bmatrix}
			\alpha_i^1,\ \alpha_i^2,\ \cdots,\ \alpha_i^{n_uT}
		\end{bmatrix}\in\mathbb{R}^{n_wT\times n_uT}
	\end{equation}
	where the vectors $\left\{\alpha_i^1,\ \alpha_i^2,\ \cdots,\ \alpha_i^{n_uT}\right\}$ form a set of unit orthogonal bases of $\mathcal{W}_i$, such that the subspace component can be simply expressed as $\mathcal{W}_i=\text{span}\left(H_i\right)$. Additionally, $H_i$ can be readily obtained through the Gram–Schmidt process by leveraging the sampled I/O data $\left(U_i^{Test},Y_i^{Test}\right)$.
\end{rema}

To further explore the geometric properties between two affine set, we present the definition of principal angles between two subspaces.
\begin{defi} \rm \label{defi-principalangles} \vspace{10pt}
	(\citep{PAAbsil2006LAIA}) For two subspaces $\mathcal{R}_1\subset\mathbb{R}^{n}$ and $\mathcal{R}_2\subset\mathbb{R}^{n}$ with $\text{dim}\left(\mathcal{R}_1\right)=\text{dim}\left(\mathcal{R}_2\right)=m,\ m\leq n$, the principal angles between $\mathcal{R}_1$ and $\mathcal{R}_2$, denoted by
	\begin{equation} \nonumber
		\Theta\left(\mathcal{R}_1,\mathcal{R}_2\right)=\begin{bmatrix}
			\theta_1,\theta_2,\cdots,\theta_{m}
		\end{bmatrix},\ \theta_k\in\left[0,\pi/2\right],\ k\in\mathbb{Z}_{m}\backslash\{0\}
	\end{equation}
	are recursively defined as
	\begin{equation} \nonumber
		s_k=\cos\left(\theta_k\right)=\max_{p\in\mathcal{R}_1}\max_{q\in\mathcal{R}_2}\langle p,q\rangle=\langle p_k,q_k\rangle
	\end{equation}
	subject to
	\begin{equation} \nonumber
		\left\|p\right\|=\left\|q\right\|=1,\ \langle p,p_i\rangle=0,\ \langle q,q_i\rangle=0,\ i\in\mathbb{Z}_{k-1}\backslash\{0\}
	\end{equation}
	Moreover, the vectors $\left\{p_1,p_2,\cdots,p_m\right\}$ and $\left\{q_1,q_2,\cdots,q_m\right\}$ are called the principal angles associated with $\mathcal{R}_1$ and $\mathcal{R}_2$.
\end{defi}
The principal angles can serve as a powerful tool to characterize the similarity between two subspaces. Based on this definition, we can further explore the similarity and similarity indexes between two admissible behaviors.

\subsection{Similarity and similarity indexes}
\label{subsec:Similarity and Similarity Indexes}
This subsection is devoted to addressing the problem P2).
Regarding the host system $\Sigma_{1,\mathbb{T}}$ with $x_1(0)=x_{10}$ and the guest system $\Sigma_{2,\mathbb{T}}$ with $x_2(0)=x_{20}$, this subsection aims at presenting the definitions of similarity and similarity indexes. Moreover, by leveraging the collected offline I/O data, a data-based criterion for verifying the similarity and a data-based strategy for calculating the similarity indexes are developed. The definition of similarity between $\mathcal{B}_{1,x_{10}}$ and $\mathcal{B}_{2,x_{20}}$ is presented as follows.
\vspace{10pt}
\begin{defi} \rm \label{defi-similarity}
	The admissible behaviors $\mathcal{B}_{1,x_{10}}$ and $\mathcal{B}_{2,x_{20}}$ are said to be similar if $\mathcal{B}_{1,x_{10}}\cap\mathcal{B}_{2,x_{20}}\neq \emptyset$.
\end{defi}
Even if $\mathcal{B}_{1,x_{10}}\cap\mathcal{B}_{2,x_{20}}\neq \emptyset$ and $\mathcal{B}_{1,x_{10}}\cap\mathcal{B}_{3,x_{30}}\neq \emptyset$, it does not necessarily follow that $\mathcal{B}_{2,x_{20}}\cap\mathcal{B}_{3,x_{30}}\neq \emptyset$. From definition \ref{defi-similarity}, when two admissible behaviors $\mathcal{B}_{1,x_{10}}$ and $\mathcal{B}_{2,x_{20}}$ are similar, there always exist some common admissible trajectories $w_{com}\in\mathcal{B}_{1,x_{10}}\cap\mathcal{B}_{2,x_{20}}$ and common behavior $\mathcal{B}_{com}:=\mathcal{B}_{1,x_{10}}\cap\mathcal{B}_{2,x_{20}}$ involving all common admissible trajectories. Taking into account the fact that the admissible behavior can be decomposed as $\mathcal{B}_{i,x_{i0}}=\text{span}\left(H_i\right)+w_i^0$, a data-based criterion for verifying the similarity can be readily derived as follows.
\vspace{10pt}
\begin{lemma}\rm \label{lemma-similaritycriterion}
	For admissible behaviors $\mathcal{B}_{1,x_{10}}$ and $\mathcal{B}_{2,x_{20}}$, let the test inputs $U_1^{Test}$ and $U_2^{Test}$ satisfy the offline test principles (\ref{eq-principle1}) and (\ref{eq-principle2}), and let $H_1$ and $H_2$ be constructed as in (\ref{eq-unitbasesofW}). Then admissible behaviors $\mathcal{B}_{1,x_{10}}$ and $\mathcal{B}_{2,x_{20}}$ are similar if and only if there exist two vectors $l_1\in\mathbb{R}^{n_uT}$ and $l_2\in\mathbb{R}^{n_uT}$ such that
	\begin{equation} \label{eq-criterion for similarity}
		\begin{bmatrix}
			H_1 & H_2
		\end{bmatrix}\begin{bmatrix}
			l_1 \\
			l_2
		\end{bmatrix}=w_2^0-w_1^0.
	\end{equation}
	Moreover, by solving the above non-homogeneous LAE, the common behavior can be expressed as
	\begin{equation} \nonumber
		\mathcal{B}_{com}=\left\{w_{com}\in\mathbb{R}^{n_wT}\left|w_{com}=H_1l_1+w_1^0\right.\right\}.
	\end{equation}
\end{lemma}
\begin{proof}
	A consequence of Lemma \ref{lemma-data-based representation} and Definition \ref{defi-similarity}.
\end{proof}

\begin{rema} \rm
	\label{rema-meaningofsimilarity}
	From Definition \ref{defi-similarity} and Lemma \ref{lemma-similaritycriterion}, it is observed that the similarity is a rather loose concept, serving only as a qualitative assessment indicator. However, it can assist to address the Problem P3) in some special cases. To be specific, once the experience provided by guest system satisfies $w_g\in\mathcal{B}_{com}$, the host system can directly adopt the successful experience $w_g$ to accomplish the same tasks. In this situation, the difference $\left\|w_g-w_h\right\|$ is equal to zero.
\end{rema}
Building upon the definition of similarity, in order to further quantitatively assess the benefits of the successful experience of the guest system  on the control of host system, the definition of similarity indexes needs to be proposed. Since the admissible behavior $\mathcal{B}_{i,x_{i0}},\ i\in\{1,2\}$ can be decomposed into a sum of subspace and offset components, from a geometric perspective, the principal angles between two subspaces $\mathcal{W}_1$ and $\mathcal{W}_2$ can serve as a powerful tool. Based on Definition \ref{defi-principalangles}, the similarity indexes between two admissible behaviors are defined as follows.
\vspace{10pt}
\begin{defi} \rm \label{defi-similarityindexes}
	For similar admissible behaviors $\mathcal{B}_{1,x_{10}}$ and $\mathcal{B}_{2,x_{20}}$, let them be decomposed as (\ref{eq-databaseddecomposition}). The similarity indexes between $\mathcal{B}_{1,x_{10}}$ and $\mathcal{B}_{2,x_{20}}$, denoted by $\textbf{SI}\left(\mathcal{B}_{1,x_{10}},\mathcal{B}_{2,x_{20}}\right)$, refer to the cosine of the principal angles between $\mathcal{W}_1$ and $\mathcal{W}_2$, that is,
	\begin{equation}\nonumber
		\textbf{SI}\left(\mathcal{B}_{1,x_{10}},\mathcal{B}_{2,x_{20}}\right):=\cos\Theta\left(\mathcal{W}_1,\mathcal{W}_2\right).
	\end{equation}
\end{defi}
\vspace{10pt}
\begin{rema}
	\rm \label{rema-properties of SI}
	The similarity indexes $\textbf{SI}\left(\cdot,\cdot\right)$ can be regarded as the function with respect to two intersecting affine sets, and $\textbf{SI}$ has the following properties:
	\begin{enumerate}
		\item[P1)] $\textbf{SI}\left(\mathcal{B}_{1,x_{10}},\mathcal{B}_{2,x_{20}}\right)$=$\textbf{SI}\left(\mathcal{B}_{2,x_{20}},\mathcal{B}_{1,x_{10}}\right)$, $\forall \mathcal{B}_{1,x_{10}}$, $\forall \mathcal{B}_{2,x_{20}}$;
		\item[P2)] $\textbf{SI}\left(\mathcal{B}_{1,x_{10}},\mathcal{B}_{1,x_{10}}\right)=1_{n_uT}^{\rm T}$, $\forall \mathcal{B}_{1,x_{10}}$;
		\item[P3)] If $\textbf{SI}\left(\mathcal{B}_{1,x_{10}},\mathcal{B}_{2,x_{20}}\right)=\textbf{SI}\left(\mathcal{B}_{1,x_{10}},\mathcal{B}_{3,x_{30}}\right)=1_{n_uT}^{\rm T}$, then $\text{\textbf{SI}}\left(\mathcal{B}_{2,x_{20}},\mathcal{B}_{3,x_{30}}\right)=1_{n_uT}^{\rm T}$, $\forall \mathcal{B}_{1,x_{10}}$, $\forall \mathcal{B}_{2,x_{20}}$, $\forall \mathcal{B}_{3,x_{30}}$.
	\end{enumerate}
\end{rema}
\vspace{10pt}
\begin{rema} \rm \label{rema-modelindependent}
	By leveraging the offline test principles (\ref{eq-principle1}) and (\ref{eq-principle2}), the decomposition in (\ref{eq-databaseddecomposition}) can be constructed from sampled data, ensuring that the calculation of similarity indexes is model-independent. Additionally, the similarity indexes between two admissible behaviors are independent on the offset components $w_{i,\text{off}}$, which can be interpreted through a geometric perspective. Since two similar admissible behaviors are essentially intersecting affine hyperplanes in Euclidean spaces, the offset components only cause translations of the corresponding affine hyperplanes and does not affect their intersection and the principal angles. Compared to the concept of similarity proposed in Definition \ref{defi-similarity}, the similarity indexes are quantitative assessment indicator. The similarity indexes closer to $1^{\rm T}_{n_uT}$ indicates that two admissible behaviors $\mathcal{B}_{1,x_{10}}$ and $\mathcal{B}_{2,x_{20}}$ are more similar.
\end{rema}

Although Definition \ref{defi-similarityindexes} presents the concept of the similarity indexes, such a definition is inefficient to calculate the similarity indexes. In order to develop an efficient calculation strategy, we define the singular value decomposition (SVD) of the matrix $H_1^{\rm T}H_2$ as
\begin{equation} \label{eq-SVD}
	H_1^{\rm T}H_2=UDV^{\rm T}
\end{equation}
where
\begin{equation} \nonumber
	D=\text{diag}\left(s_1,s_2,\cdots,s_{n_uT}\right),\ s_1\geq s_2\geq \cdots\geq s_{n_uT}>0.
\end{equation}
Through the designed offline I/O test principles (\ref{eq-principle1}) and (\ref{eq-principle2}), a data-based strategy can be proposed to efficiently calculate the similarity indexes, which is demonstrated in the following theorem.
\vspace{10pt}
\begin{theorem} \rm \label{theorem-calculationofsimilarityindexes}
	For admissible behaviors $\mathcal{B}_{1,x_{10}}$ and $\mathcal{B}_{2,x_{20}}$, let
	\begin{enumerate}
		\item The test inputs $U_1^{Test}$ and $U_2^{Test}$ fulfill the offline test principles (\ref{eq-principle1}) and (\ref{eq-principle2});
		\item The matrices $H_1$ and $H_2$ be constructed as in (\ref{eq-unitbasesofW});
		\item The SVD of $H_1^{\rm T}H_2$ be given as (\ref{eq-SVD}) where $D=\text{diag}\left(s_1,s_2,\cdots,s_{n_uT}\right)$ and $s_1\geq s_2\geq \cdots\geq s_{n_uT}>0$.
	\end{enumerate}
	If (\ref{eq-criterion for similarity}) is solvable,	then the similarity indexes between $\mathcal{B}_{1,x_{10}}$ and $\mathcal{B}_{2,x_{20}}$ can be calculated as
	\begin{equation} \label{eq-calculation of similarity indexes}
		\textbf{SI}\left(\mathcal{B}_{1,x_{10}},\mathcal{B}_{2,x_{20}}\right)=\begin{bmatrix}
			s_1,\ s_2,\ \cdots,\ s_{n_uT}
		\end{bmatrix}.
	\end{equation}
	Moreover, the principal vectors associated with $\mathcal{W}_1$ and $\mathcal{W}_2$ are given by $H_1U$ and $H_2V$.
\end{theorem}
\begin{proof}
	Detailed proof is given in Appendix A.
\end{proof}
After presenting the definition of similarity indexes and calculating them from sampled data, we can pay our attention back to Problem P3), which is addressed in Section \ref{sec:Similarity-Based Learning Control Strategy}.

\section{Similarity-based learning control framework} \label{sec:Similarity-Based Learning Control Strategy}
In this section, a similarity-based learning control framework is proposed by leveraging the sampled I/O data to address Problem P3). We suppose that, through some powerful control strategies, the guest system $\Sigma_{2,\mathbb{T}}$ has already accomplished its tasks and learned the admissible trajectory $w_g$. The core idea of the similarity-based learning control framework lies in that when the host system $\Sigma_{1,\mathbb{T}}$ is confronted with the same tasks, the successful experience of the guest system can provide helpful guidance.
Moreover, the benefits of the successful experience of the guest system to the host system can be quantitively assessed via the similarity indexes introduced in Section \ref{sec:Data-Based Verification For Similarity and Similarity Indexes},.

Specifically, as we revisit Problem P3), it is evident that the to-be-sought $w_h$ is essentially the orthogonal projection of $w_g$ onto $\mathcal{B}_{1,x_{10}}$. Existing learning-based control strategies depend on the model information of $\Sigma_{1,\mathbb{T}}$, adjusting the controller parameters through repetitive trial-and-error process to ultimately find $w_h$, which minimizes the difference $\left\|w_g-w_h\right\|$. With respect to the mechanism of the existing learning-based control strategies, an illustrative example in the $3$-dimensional Euclidean space $\mathbb{R}^3$ is depicted in Fig. \ref{fig-similarlearning1}. 

In contrast, the similarity-based learning control framework aims to directly obtain $w_h$ via projection techniques by employing the similarity indexes and the successful experience of the guest system. Consequently, the trial-and-error processes are no longer needed. Likewise, an illustrative example is depicted in Fig. \ref{fig-similarlearning2}, where $w_h$ can be efficiently calculated by exploiting the similarity indexes $\cos\Phi$ and $w_g$, ensuring that $\left\|w_h-w_g\right\|$ is minimized.
\begin{figure}[!ht]
	\centering
	\includegraphics[width=0.8\columnwidth]{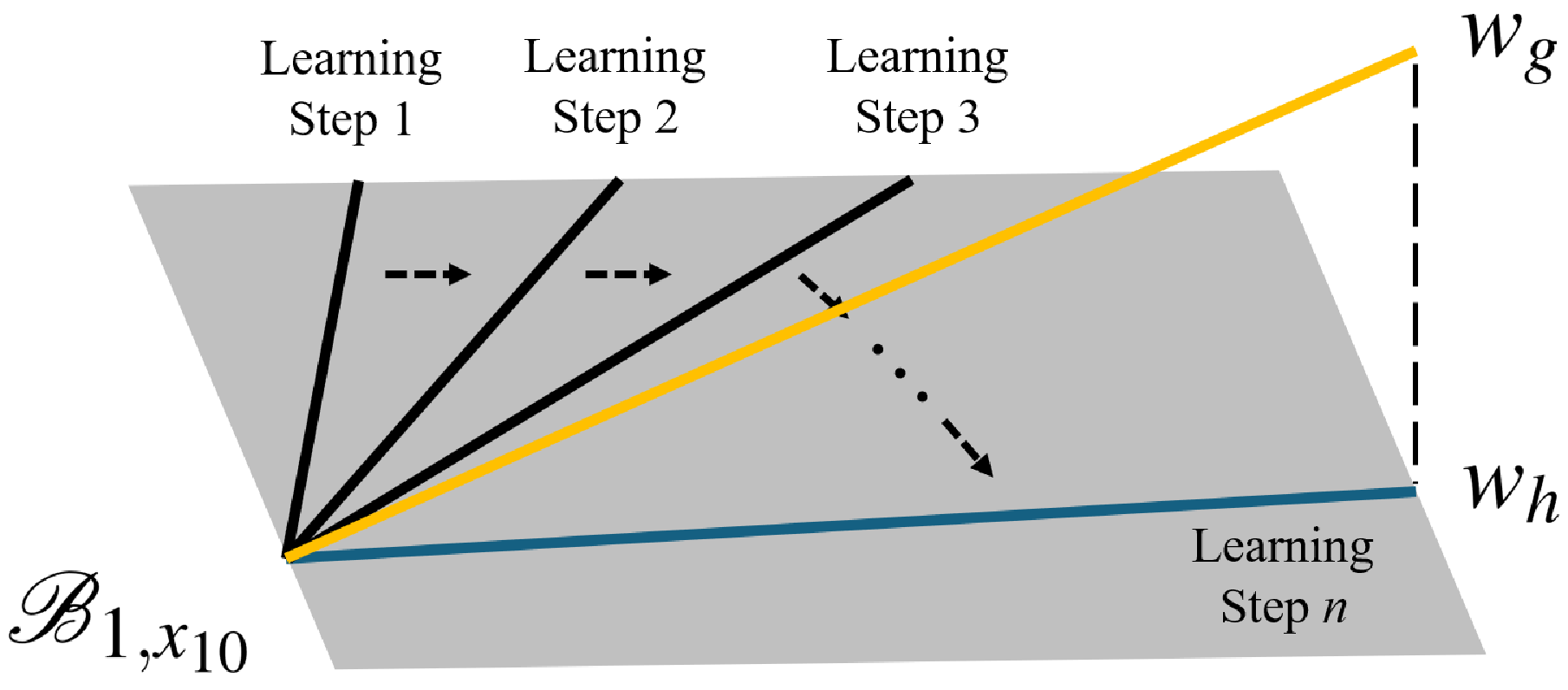}
	\caption{Existing learning-based control strategies for seeking $w_h$.}
	\label{fig-similarlearning1}
\end{figure}
\begin{figure}[!ht]
	\centering
	\includegraphics[width=0.8\columnwidth]{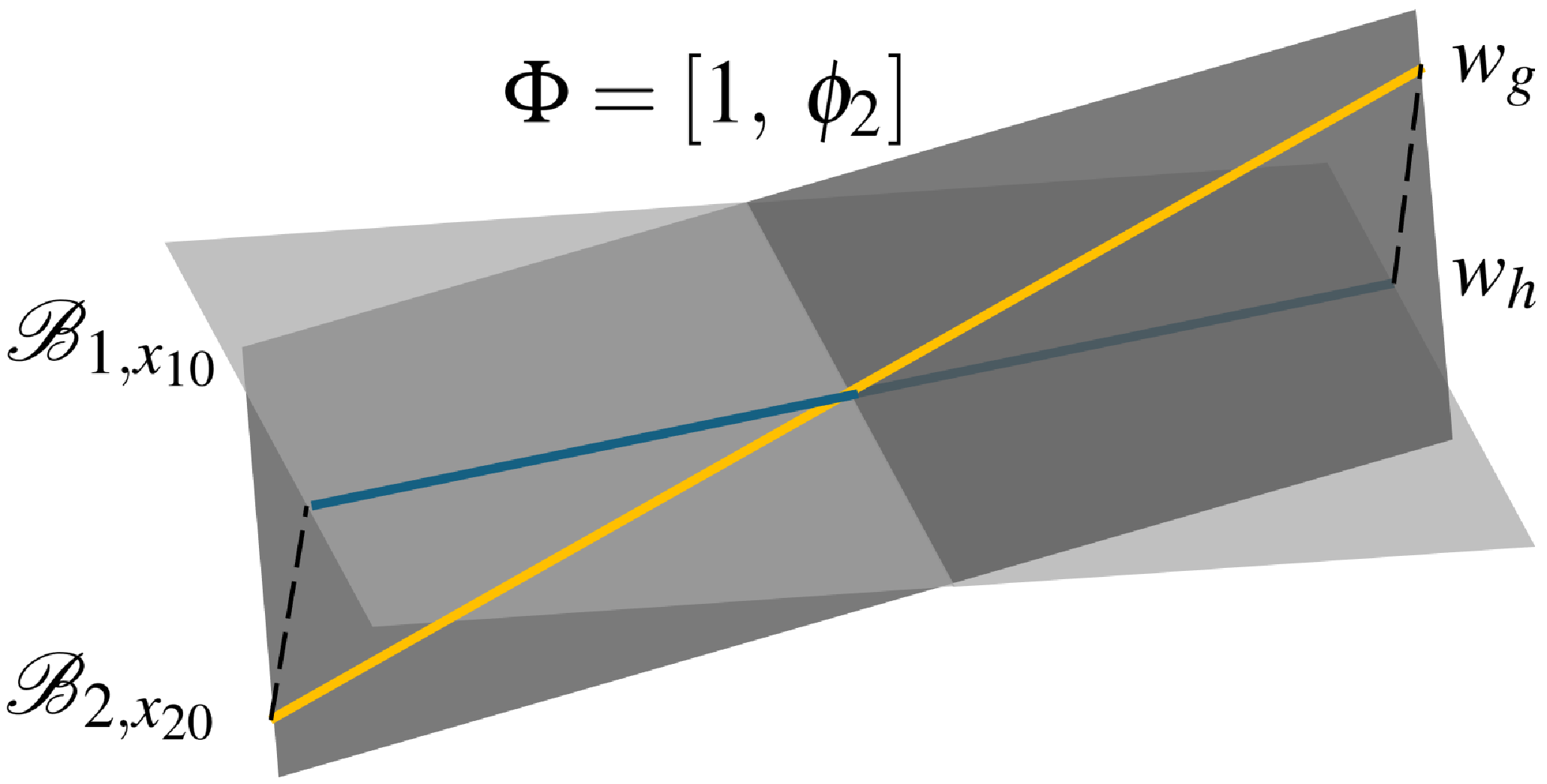}
	\caption{Similarity-based learning control strategy for seeking $w_h$.}
	\label{fig-similarlearning2}
\end{figure}
To present the similarity-based learning control framework more precisely, the orthogonal projection operator onto the subspace $\mathcal{W}_1$ is denoted by $P_{\mathcal{W}_1}\left(\cdot\right):\mathbb{R}^{n_wT}\rightarrow\mathcal{W}_1$. Correspondingly, the orthogonal projection operator onto the admissible behavior $\mathcal{B}_{1,x_{10}}$ is denoted by $P_{\mathcal{B}_{1,x_{10}}}\left(\cdot\right):\mathbb{R}^{n_wT}\rightarrow\mathcal{B}_{1,x_{10}}$. That is, for any $w_g\in\mathbb{R}^{n_wT}$, $P_{\mathcal{B}_{1,x_{10}}}\left(w_g\right)$ refers to its orthogonal projection onto $\mathcal{B}_{1,x_{10}}$ that minimizes the difference $\left\|w_g-P_{\mathcal{B}_{1,x_{10}}}\left(w_g\right)\right\|$. Before presenting the similarity-based learning control framework, the following lemma is introduced as the preliminary.
\vspace{10pt}
\begin{lemma} \rm \label{lemma-projection}
	(\citep{JPlesnik2007LAIA})
	For the admissible behavior $\mathcal{B}_{1,x_{10}}$, let the test input $U_1^{Test}$ satisfy the offline test principles (\ref{eq-principle1}) and (\ref{eq-principle2}), and let $H_1$ be constructed as in (\ref{eq-unitbasesofW}). Then for all $x\in\mathbb{R}^{n_wT}$, the orthogonal projection onto $\mathcal{B}_{1,x_{10}}$ can be calculated by
	\begin{equation} \nonumber
		P_{\mathcal{B}_{1,x_{10}}}(x)=w_{1}^0+P_{\text{span}\left(H_1\right)}\left(x-w_{1}^0\right),\ \forall x\in\mathbb{R}^{n_wT}.
	\end{equation}
\end{lemma}
By leveraging the offline sampled data, Lemma \ref{lemma-projection} calculates the orthogonal projection onto the admissible behavior via investigating another orthogonal projection onto the associated subspace component. In comparison to existing learning-based control strategies, the main superiority of the proposed similarity-based learning control framework lies in the fact that the term $P_{\text{span}\left(H_1\right)}\left(x-w_1^0\right)$ can be efficiently obtained by exploiting the similarity indexes and projection techniques, which is demonstrated as follows.
\vspace{10pt}
\begin{theorem} \rm \label{theorem-similaritybasedlearning}
	For admissible behaviors $\mathcal{B}_{1,x_{10}}$ and $\mathcal{B}_{2,x_{20}}$, let
	\begin{enumerate}
		\item[C1)] \label{condition1} The test inputs $U_1^{Test}$ and $U_2^{\text{Test}}$ fulfill the offline test principles (\ref{eq-principle1}) and (\ref{eq-principle2});
		\item[C2)] \label{condition2} The matrix $H_1$ and $H_2$ be constructed as in (\ref{eq-unitbasesofW});
		\item[C3)] \label{condition3} There exist $l_1$ and $l_2$ such that (\ref{eq-criterion for similarity}) holds;
		\item[C4)] \label{condition4} The SVD of $H_1^{\rm T}H_2$ be given as (\ref{eq-SVD}) where $D=\text{diag}\left(s_1,s_2,\cdots,s_{n_uT}\right)$ and $s_1\geq s_2\geq \cdots\geq s_{n_uT}>0$.
	\end{enumerate}
	For the learned admissible trajectory given by $w_g\in\mathcal{B}_{2,x_{20}}$, the optimal admissible trajectory $w_h\in\mathcal{B}_{1,x_{10}}$ can be calculated according to
	\begin{equation} \label{eq-calculation of w_h}
		w_h=H_1UD\overline{g}+P_{\text{span}\left(H_1\right)}\left(w_{2}^0-w_{1}^0\right)+w_1^0
	\end{equation}
	where $\overline{g}$ satisfies
	\begin{equation} \label{eq-representation of w_g}
		w_g=H_2V\overline{g}+w_{2}^0.
	\end{equation}
	In this situation, the difference $\left\|w_h-w_g\right\|$ is minimized, or equivalently, Problem P3) is addressed.
\end{theorem}
\begin{proof}
	Detailed proof is given in Appendix B.
\end{proof}

\begin{rema} \rm
	The similarity-based learning control framework proposed in Theorem \ref{theorem-similaritybasedlearning} can provide  an innovative perspective on learning-based control. When seeking the optimal trajectory $w_h\in\mathcal{B}_{1,x_{10}}$, we no longer require to repeatedly execute certain learning-based control strategies for $\Sigma_{1,\mathbb{T}}$. Alternatively, we can directly obtain the optimal trajectory $w_h$ by leveraging the successful experience of guest systems. From Theorem \ref{theorem-similaritybasedlearning}, it can be immediately observed that the closer $\textbf{SI}\left(\mathcal{B}_{1,x_{10}},\mathcal{B}_{2,x_{20}}\right)$ is to $1^{\rm T}_{n_uT}$, the smaller the difference $\Vert w_h-w_g \Vert$. Particularly, in the scenarios where $\textbf{SI}\left(\mathcal{B}_{1,x_{10}},\mathcal{B}_{2,x_{20}}\right)=1_{n_uT}^{\rm T}$, applying the control strategy in Theorem \ref{theorem-similaritybasedlearning} does not lead to any learning errors. That is, $w_h=w_g$ holds in such scenarios.
\end{rema}
\vspace{10pt}
\begin{rema}\rm
	Moreover, in those low-similarity scenarios, the similarity indexes between $\mathcal{B}_{1,x_{10}}$ and $\mathcal{B}_{2,x_{20}}$ can be compensated by interconnecting the host system with another auxiliary system (see, e.g., \cite{JCWillems2007CSM}, for more details). Through the compensation, it is expected that 
	\begin{equation} \nonumber
		\left\|1_{n_uT}^{\rm T}-\textbf{SI}\left(\mathcal{B}_{1,x_{10}},\mathcal{B}_{2,x_{20}}\right)\right\|>\left\|1_{n_uT}^{\rm T}-\textbf{SI}\left(\mathcal{B}^{\text{comp}}_{1,x_{10}},\mathcal{B}_{2,x_{20}}\right)\right\|
	\end{equation}
	where $\mathcal{B}^{\text{comp}}_{1,x_{10}}$ is the compensated admissible behavior. That is, we hope $\mathcal{B}^{\text{comp}}_{1,x_{10}}$ can be more similar to $\mathcal{B}_{2,x_{20}}$ via the compensation.
	Based on this consideration, satisfied learning performances can almost always be achieved when applying similarity-based learning control.
\end{rema}
\vspace{10pt}
\begin{rema} \rm
	Just like humans need to absorb a wide range of learning experiences from others, an increase in the number of experience provider will improve the control performance of the similarity-based learning. This is because, as the number of guest systems increases, there always exists a guest system whose admissible behavior shares more similarity with that of the host system. By adopting the successful experience of the ``most similar'' guest system, the similarity-based learning control framework can eventually result in better control performances.
\end{rema}

Building upon the previously proposed results, the procedure of executing the similarity-based learning control framework can be summarized in Algorithm \ref{alg-similarityindexes}.
\begin{algorithm}[htb]
	\caption{ Similarity-based learning control.}
	\label{alg-similarityindexes}
	\begin{algorithmic}[1]
		\REQUIRE
		\STATE Apply the test inputs $U_i^{Test}$ satisfying (\ref{eq-principle1}) and (\ref{eq-principle2}) to $\Sigma_{i,\mathbb{T}}$;
		\STATE Collect the offline I/O data in $\left(U_i^{Test},Y_i^{Test}\right)$;
		\STATE Obtain the data-based decomposition (\ref{eq-databaseddecomposition}) and (\ref{eq-unitbasesofW}) through Gram-Schmidt process.
		\ENSURE
		\STATE Check the similarity between $\mathcal{B}_{1,x_{10}}$ and $\mathcal{B}_{2,x_{20}}$ via (\ref{eq-criterion for similarity})\\
		\qquad \textbf{If} (\ref{eq-criterion for similarity}) is sovable, \textbf{then} go to step 5;\\
		\qquad \textbf{Else}, quit this algorithm.
		\STATE Calculate the SVD as $H_1^{\rm T}H_2=UDV^{\rm T}$;
		\STATE Obtain the similarity indexes as (\ref{eq-calculation of similarity indexes});
		\STATE Obtain the principal vectors as $H_1U$ and $H_2V$;
		\STATE Calculate the required $w_h$ via (\ref{eq-calculation of w_h}) and (\ref{eq-representation of w_g}).
	\end{algorithmic}
\end{algorithm}

\section{Simulation examples} \label{sec:Simulation Example}
For the illustration of the proposed similarity-based learning control frameworks, simulation examples are presented in this section. We provide a numerical example and simulation tests on the mobile robots simultaneously.
\vspace{10pt}
\begin{example} \rm
	Consider two heteronogeous discrete-time linear systems in the form of (\ref{statespacemodel}), and their model matrices are given as follows:
	\begin{equation} \nonumber
		\begin{aligned}
			A_1(t)&=\begin{bmatrix}
				0.05t & 1 & 0\\
				0 & 0.05t & 1\\
				-0.09 & -0.60 & -1.40+0.05t
			\end{bmatrix},\\
			A_2(t)&=\begin{bmatrix}
				0.05t & 1 & 0\\
				0 & 0.05t & 1\\
				-0.08 & -0.66 & -1.50+0.05t
			\end{bmatrix},\\
			B_1(t)&=B_2(t)=\begin{bmatrix}
				6\\
				0\\
				0.50
			\end{bmatrix},\ C_1(t)=C_2(t)=\begin{bmatrix}
				2\\
				1\\
				0
			\end{bmatrix}^{\rm T},\\
			D_1(t)&=D_2(t)=0,\ x_1(0)=\begin{bmatrix}
				0 \\ 0 \\ 1.02
			\end{bmatrix},\ x_2(0)=\begin{bmatrix}
				0 \\ 0 \\ 1
			\end{bmatrix}.
		\end{aligned}
	\end{equation}
	Here, we present the model knowledge solely for clear illustrating the simulation settings, and it will not be utilized for the design and analysis. Of note is that this type of difference often arises in scenarios where there exist uncertainties between the host system $\Sigma_{1,\mathbb{T}}$ and guest system $\Sigma_{2,\mathbb{T}}$. Let the host and guest systems be given a consistent output tracking task over the time duration $\mathbb{Z}_{34}$, with the reference output set as
	\begin{equation} \nonumber
		y_d(t)=e^{-0.1t}\sin(\frac{\pi}{5}t),\ \forall t\in\mathbb{Z}_{34}.
	\end{equation}
	Owing to the absence of model knowledge, offline I/O tests are needed to collect a sufficient number of admissible trajectories. The offline I/O tests need to be executed for at least 36 times, and the test inputs are designed as
	\begin{equation} \nonumber
		\mathbf{u}_i^0=0_{35},\ \begin{bmatrix}
			\mathbf{u}_i^1,\ \mathbf{u}_i^2,\ \cdots,\mathbf{u}_i^{35}
		\end{bmatrix}=I_{35}.
	\end{equation}
	With the designed test inputs, the proposed offline test principles (\ref{eq-principle1}) and (\ref{eq-principle2}) are satisfied, then the data-based representation (\ref{eq-data-basedrepresentation}) can be constructed, and the admissible behaviors $\mathcal{B}_{1,x_{10}}$ and $\mathcal{B}_{2,x_{20}}$ can be decomposed by leveraging (\ref{eq-databaseddecomposition}).
	To address the output tracking problem of the guest system $\Sigma_{2,\mathbb{T}}$, iterative learning control (ILC) that is a learning-based strategy can serve as a powerful tool. After 300 iterations of the algorithm, the tracking problem of the guest system $\Sigma_{2,\mathbb{T}}$ is perfectly addressed. The output and input of the guest system, denoted as $y_g(t)-\text{ILC}$ and $u_g(t)-\text{ILC}$, respectively, are depicted in Fig. \ref{fig-system12}.
	\begin{figure}[!ht]
		\centering
		{\includegraphics[width=0.9\columnwidth]{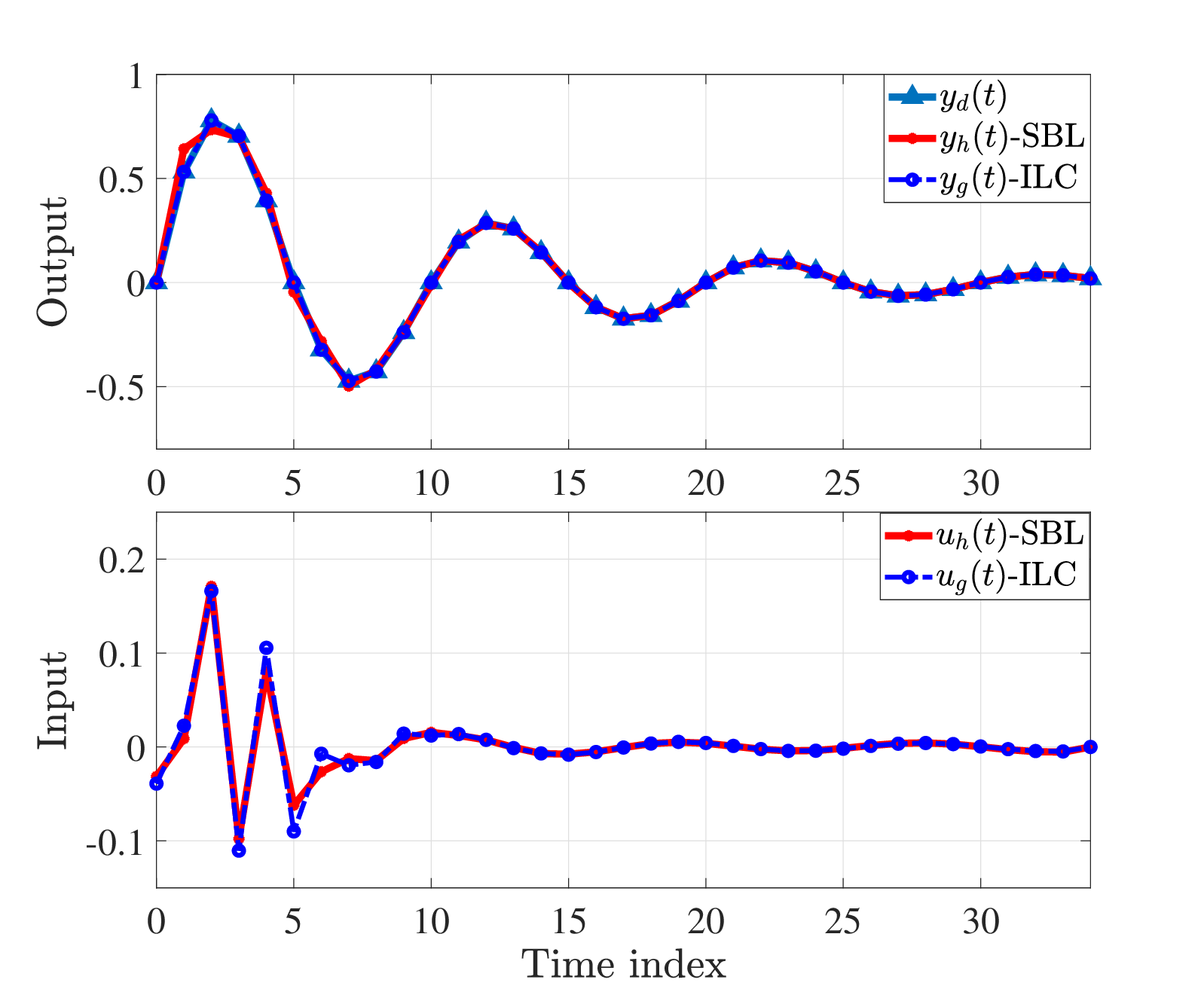}}
		\caption{Outputs of the system $\Sigma_{1,\mathbb{T}}$ and $\Sigma_{2,\mathbb{T}}$ for reference $y_d(t)$.}
		\label{fig-system12}
	\end{figure}
	With the obtained admissible trajectory $w_g\in\mathcal{B}_{2,x_{20}}$, the tracking issue of the host system no longer needs to resort to ILC, which depends on repetitive learning and trial-and-error. By exploiting the proposed similarity-based learning control framework, the required $w_h\in\mathcal{B}_{1,x_{10}}$ can be obtained through Theorem \ref{theorem-similaritybasedlearning}. The learned output and input of the host system, denoted by $y_h(t)$-SBL and $u_h(t)$-SBL, respectively, are depicted in Fig. \ref{fig-system12}. From Fig. \ref{fig-system12}, it can be observed that the similarity-based learning control framework achieves satisfied learning performances, and the tracking issue of the host system can be directly addressed.
	
	As a comparison, another guest system, denoted by $\Sigma_{3,\mathbb{T}}$, that is less similar with the host system $\Sigma_{1,\mathbb{T}}$ is also provided. The model knowledge of $\Sigma_{3,\mathbb{T}}$ is presented as follows:
	\begin{equation}\nonumber
		\begin{aligned}
			A_3(t)&=\begin{bmatrix}
				0.05t & 1 & 0 \\
				0 & 0.05t & 1 \\
				-0.20 & -0.20 & -1.3+0.05t
			\end{bmatrix},\\
			B_3(t)&=\begin{bmatrix}
				6\\
				0\\
				0.50
			\end{bmatrix},\ C_3(t)=\begin{bmatrix}
				2 \\
				1\\
				0
			\end{bmatrix}^{\rm T},\ D_3(t)=0.
		\end{aligned}
	\end{equation}
	The initial state of $\Sigma_{3,\mathbb{T}}$ is set as $x_3(0)=\begin{bmatrix}
		0.2,\ 0,\ 1
	\end{bmatrix}^{\rm T}$, and the tracking task for the reference output $y_d(t)$ are taken into account again. By applying ILC to $\Sigma_{3,\mathbb{T}}$, the tracking issue of $\Sigma_{3,\mathbb{T}}$ can be addressed. The output and input of $\Sigma_{3,\mathbb{T}}$, denoted as $y'_g(t)$-ILC and $u'_g(t)$-ILC, respectively, are depicted in Fig. \ref{fig-system13}. After obatining the admissible trajectory $w_{g}'\in\mathcal{B}_{3,x_{30}}$, the tracking problem of the host system can be directly addressed through the similarity-based learning control framework. The learned output and input of the system $\Sigma_{1,\mathbb{T}}$, denoted as $y'_h(t)$-SBL anf $u'_h(t)$-SBL, respectively, are depicted in Fig. \ref{fig-system13}. Since the guest system $\Sigma_{3,\mathbb{T}}$ is less similar with the host system $\Sigma_{1,\mathbb{T}}$, the performance brought by the similarity-based learning control is degraded.
	\begin{figure}[!ht]
		\centering
		{\includegraphics[width=0.9\columnwidth]{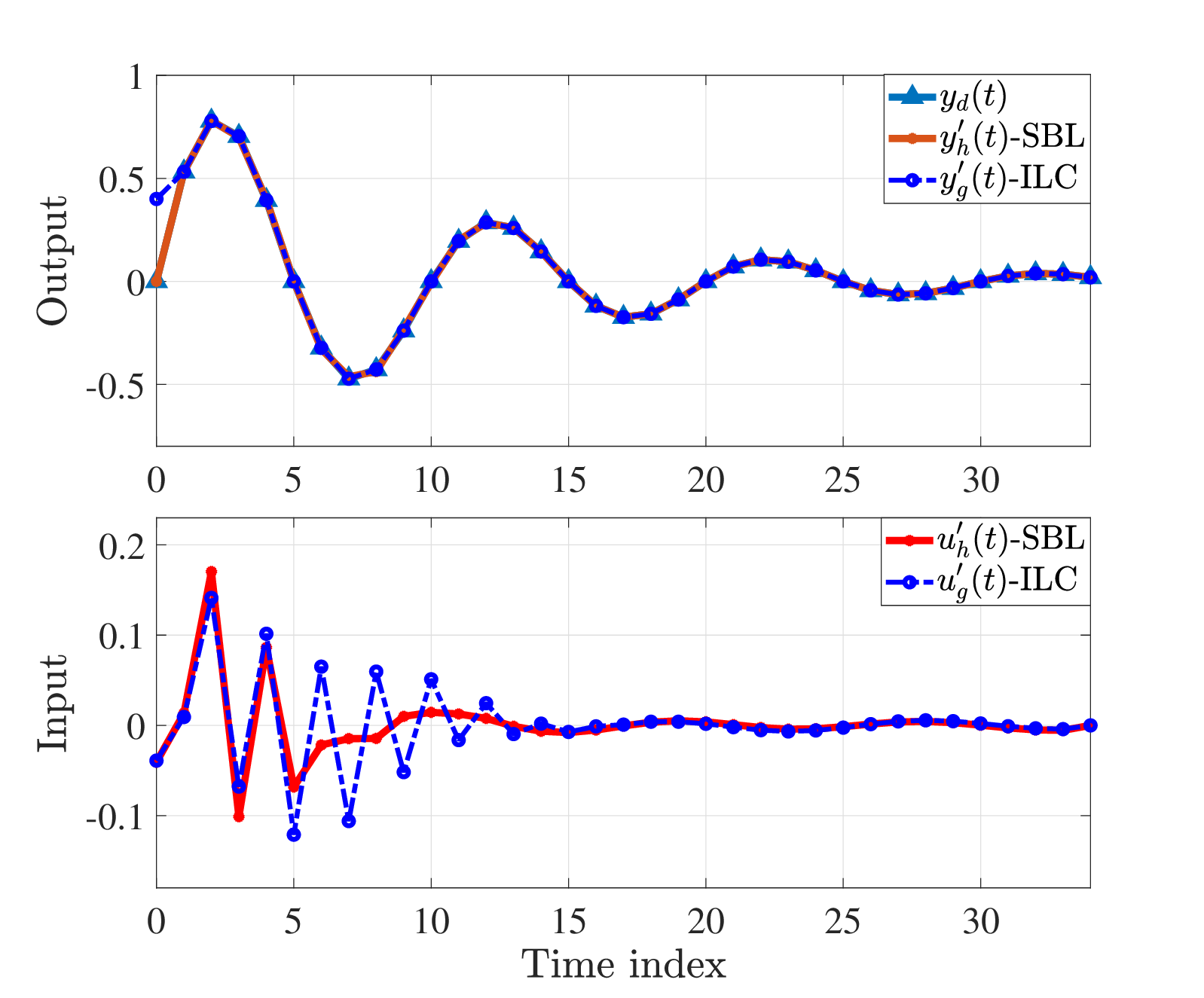}}
		\caption{Outputs of the system $\Sigma_{1,\mathbb{T}}$ and $\Sigma_{3,\mathbb{T}}$ for reference $y_d(t)$.}
		\label{fig-system13}
	\end{figure}
\end{example}

\vspace{10pt}
\begin{example}
	\rm Consider the a class of mobile robots equipped with two independent driving wheels (\citep{ZZhuang2023TSMC}), whose physical models are illustrated in Fig. \ref{fig-mobilerobots}.
	\begin{figure}[!ht]
		\centering
		{\includegraphics[width=0.7\columnwidth]{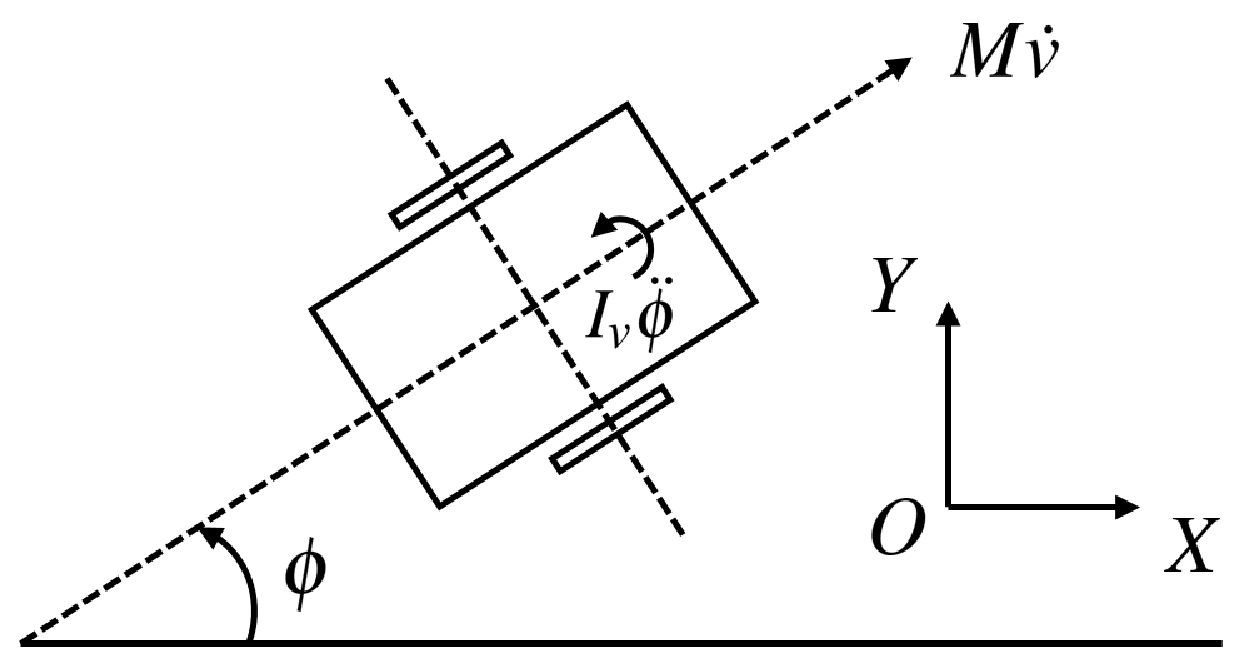}}
		\caption{A class of mobile robots with two driving wheels.}
		\label{fig-mobilerobots}
	\end{figure}
	The symbols $v$, $\phi$, $u_r$, and $u_l$ represents the velocity, azimuth, right driving input, and left driving input of the mobile robot, repectively. Let the state, input, and output of the robots be defined as $x=\begin{bmatrix}
		v,\ \phi,\ \dot{\phi}
	\end{bmatrix}^{\rm T}$,
	$u=\begin{bmatrix}
		u_r,\ u_l
	\end{bmatrix}^{\rm T}$,
	and $y=\begin{bmatrix}
		v,\ \phi
	\end{bmatrix}^{\rm T}$, respectively, and let the sampling time be $T_s=0.05s$. Through the discretization and linearization techniques, the dynamics of the mobile robots can be described by the state space representation as
	\begin{equation} \label{eq-robots}
		R_i:\left\{\begin{aligned}
			x_i(n+1)&=A_ix_i(n)+B_iu_i(n)\\
			y_i(n)&=C_ix_i(n)
		\end{aligned}\right. .
	\end{equation}
	The symbol $n\in\mathbb{Z}_+$ refers to the sampling points, thus the time interval between two adjacent sampling points is $T_s$.
	For the host robot $R_1$ and guest robot $R_2$ whose dynamics are represented by (\ref{eq-robots}), their model parameters are given as
	\begin{equation} \nonumber
		\begin{aligned}
			A_1&=\begin{bmatrix}
				1.0100 & 0 & 0\\
				0 & 1 & 0.0520\\
				0 & 0 & 1.0100
			\end{bmatrix},\
			A_2=\begin{bmatrix}
				0.9975 & 0 & 0\\
				0 & 1 & 0.0499\\
				0 & 0 & 0.9955
			\end{bmatrix},\\
			B_1&=\begin{bmatrix}
				0.0130 & 0.0130\\
				-0.0025 & -0.0050\\
				-0.0850 & -0.1700
			\end{bmatrix},\ B_2=\begin{bmatrix}
			0.0125 & 0.0125\\
			-0.0021 & -0.0042\\
			-0.0833 & -0.1666
			\end{bmatrix},\\ C_1&=C_2=\begin{bmatrix}
				1 & 0 \\
				0 & 1 \\
				0 & 0
			\end{bmatrix},\ x_1(0)=\begin{bmatrix}
				3 \\ 0 \\ 0
			\end{bmatrix},\ x_2(0)=\begin{bmatrix}
				3.02 \\ 0 \\ 1
			\end{bmatrix}.
		\end{aligned}
	\end{equation}
	We provide the model parameters solely to illustrate the simulation settings clearly. In the scenarios where the model information is not available, we can still obtain the data-based representation for the admissible behavior by designing appropriate offline test principles, as discussed in Lemma \ref{lemma-data-based representation}.
	
	Two mobile robots are assigned the same task, which is to move along a preplanned circular path within the time duration $\mathbb{T}=\left[0,4\right]$s. The circular path is specified by the velocity and azimuth references of the mobile robots. Specifically, within the time duration $\mathbb{T}=\left[0,4\right]s$, the reference trajectories for velocity and azimuth are defined as
	\begin{equation} \nonumber
		\begin{aligned}
			y_{d,v}(n)&=3\left(\text{m/s}\right), \ \forall n\in\mathbb{Z}_{79},\\
			y_{d,\phi}(n)&=\left\{
			\begin{aligned}
				&0 \left(\text{rad}\right),\ n\in\mathbb{Z}_{10}\\
				&-0.6875(n-11) \left(\text{rad}\right),\ n\in\mathbb{Z}_{79}\backslash\mathbb{Z}_{10}
			\end{aligned} \right. .
		\end{aligned}
	\end{equation}
	Therefore, the equivalent objective is to track the reference $y_d=\begin{bmatrix}
		y_{d,v}(n),\ y_{d,\phi}(n)
	\end{bmatrix}^{\rm T}$ over the specific time duration $\mathbb{T}$. For the guest mobile robot, ILC can efficiently address the tracking problems. After 200 iterations, the tracking performances brought by ILC are shown in Fig. \ref{fig-velocityandazimuth}, where the learned velocity and azimuth are denoted as $v_g(n)$-ILC and $\phi_g(n)$-ILC, respectively.
	\begin{figure}[!ht]
		\centering
		{\includegraphics[width=0.9\columnwidth]{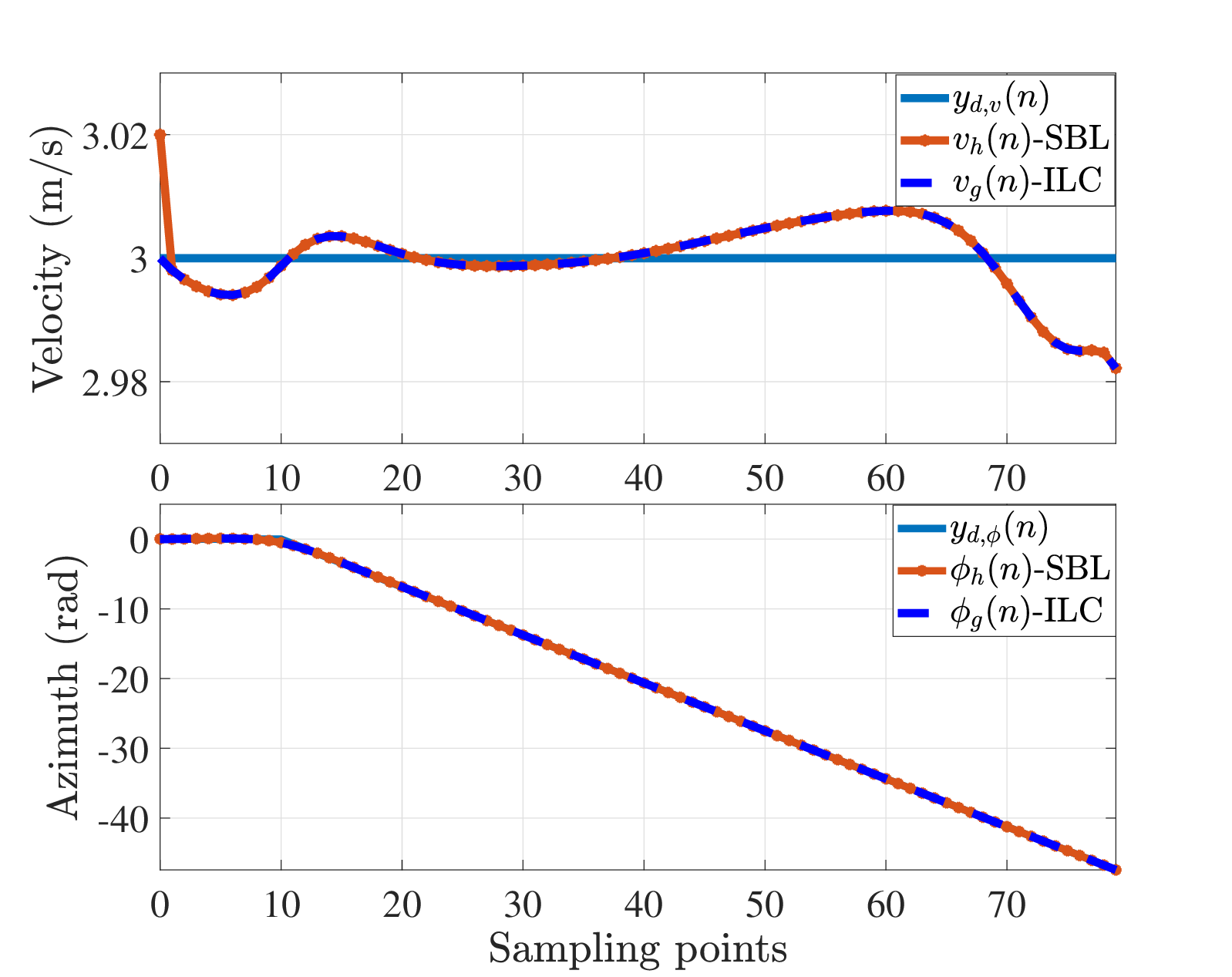}}
		\caption{Velocities and azimuths of mobile robots $R_1$ and $R_2$.}
		\label{fig-velocityandazimuth}
	\end{figure}
	As a result, the guest mobile robots gradually reaches the predefined circular trajectory. The learning process of the guest mobile robots is shown in Fig. \ref{fig-learningprocess}.
	\begin{figure}[!ht]
		\centering
		{\includegraphics[width=0.9\columnwidth]{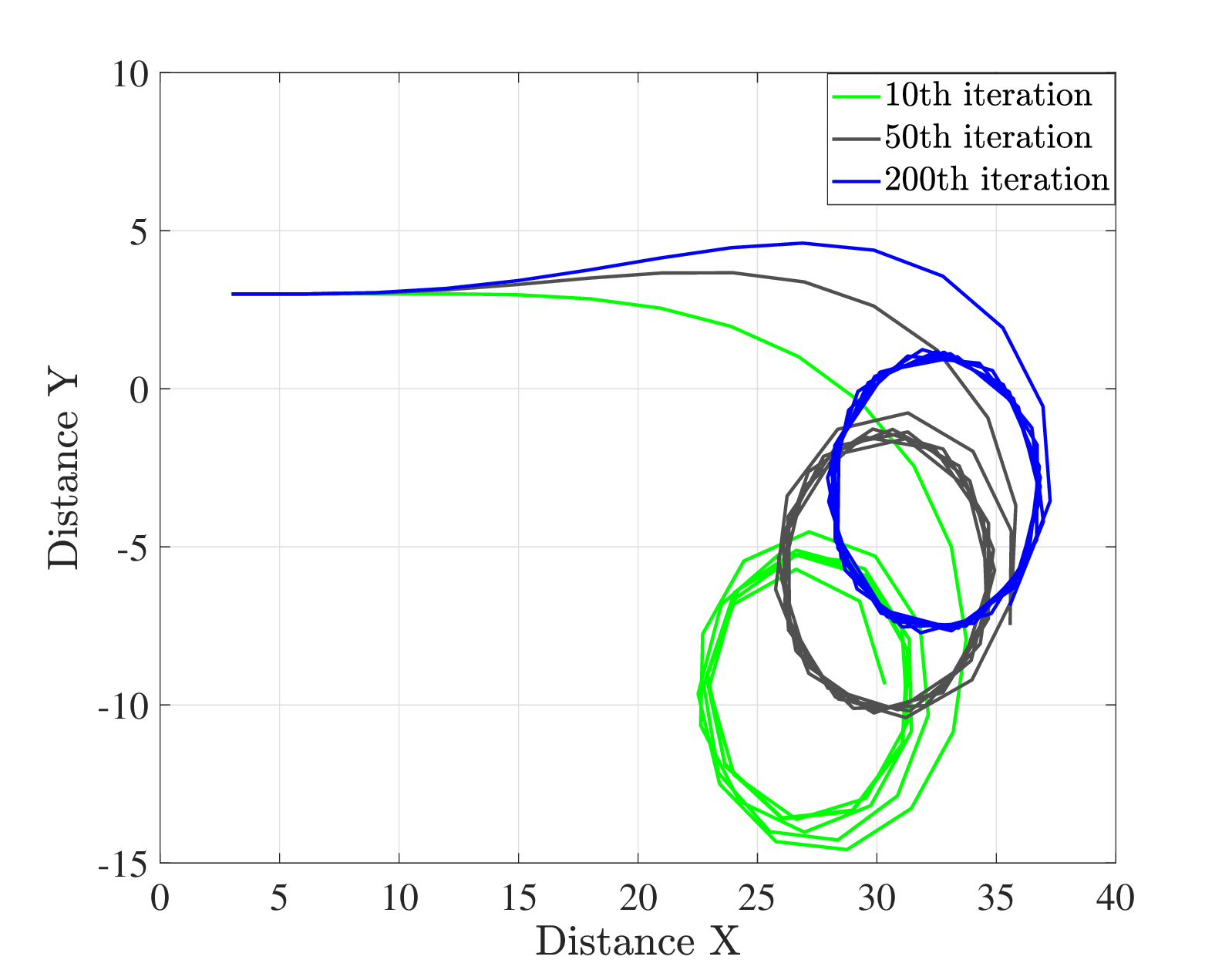}}
		\caption{Learning process of the guest mobile robots.}
		\label{fig-learningprocess}
	\end{figure}
	After 200 iterations of the algorithm, it can be observed that the mobile robot is able to proceed along the predefined circular trajectory.
	
	For the host mobile robot, it no longer relies on repetitive trial-and-error processes, but rather directly utilizes the successful experiences of the guest system to complete the control task. By leveraging the proposed similarity-based learning control framework, the learned velocity and azimuth of host mobile robot, denoted by $v_h(n)$-SBL and $\phi_h(n)$-SBL, respectively, are depicted in Fig. \ref{fig-velocityandazimuth}. Consequently, the learned path of the host mobile robot is shown in Fig. \ref{fig-twocircle}. From Fig. 9, it can be concluded that the host mobile robot can move along the preplanned circular path by leveraging the successful experience of the guest mobile robots, and the effectiveness of the proposed similarity-based learning control framework is verified.
	\begin{figure}[!ht]
		\centering
		{\includegraphics[width=0.9\columnwidth]{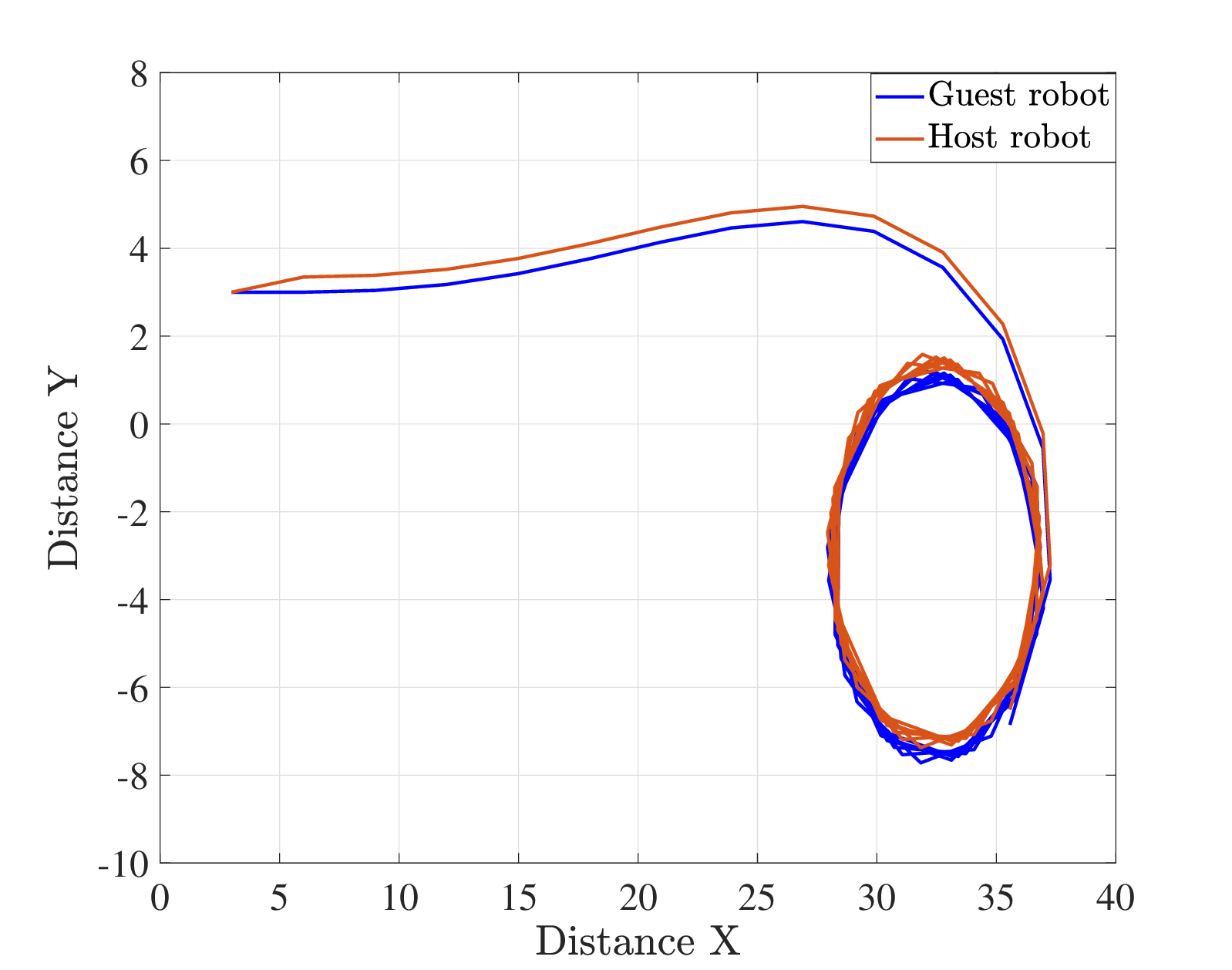}}
		\caption{Learned paths of mobile robots $R_1$ and $R_2$.}
		\label{fig-twocircle}
	\end{figure}
\end{example}
\section{Conclusions} \label{sec:Conclusions}
In this paper, we have innovatively proposed the definitions of similarity and similarity indexes between admissible behaviors, based on which a similarity-based learning control framework has been further developed. Owing to the absence of model knowledge, appropriate offline I/O test principles have been designed, based on which the admissible behaviors of LTV systems have been reconstructed from sampled data. By exploiting the sampled data, a data-based criterion for verifying the similarity and a data-based strategy for calculating the similarity indexes have been developed. Building upon the calculated similarity indexes and projection techniques, a similarity-based learning control framework has been developed by exploiting the sampled data. Consequently, the host system has accomplished the same tasks by leveraging the successful experience of the guest system, without repeatedly resorting to any existing learning-based control strategies.

\bibliographystyle{plainnat}        
\bibliography{myref}

\begin{thebibliography}{19}
\providecommand{\natexlab}[1]{#1}
\providecommand{\url}[1]{\texttt{#1}}
\expandafter\ifx\csname urlstyle\endcsname\relax
  \providecommand{\doi}[1]{doi: #1}\else
  \providecommand{\doi}{doi: \begingroup \urlstyle{rm}\Url}\fi

\bibitem[Absil et~al.(2006)Absil, Edelman, and Koev]{PAAbsil2006LAIA}
P.~A. Absil, A.~Edelman, and P.~Koev.
\newblock On the largest principal angle between random subspaces.
\newblock \emph{Linear Algebra and Its Applications}, 414\penalty0
  (1):\penalty0 288--294, 2006.

\bibitem[Arimoto et~al.(1984)Arimoto, Kawamura, and Miyazaki]{SArimoto1984JRS}
S.~Arimoto, S.~Kawamura, and F.~Miyazaki.
\newblock Bettering operation of robots by learning.
\newblock \emph{Journal of Robotic Systems}, 1\penalty0 (2):\penalty0 123--140,
  1984.

\bibitem[Bertolini et~al.(2021)Bertolini, Mezzogori, Neroni, and
  Zammori]{MBertolini2021ESA}
M.~Bertolini, D.~Mezzogori, M.~Neroni, and F.~Zammori.
\newblock Machine learning for industrial applications: {A} comprehensive
  literature review.
\newblock \emph{Expert Systems with Applications}, 175:\penalty0 114820, 2021.

\bibitem[Boyd and Vandenberghe(2004)]{SBoyd2004convex}
S.~P. Boyd and L.~Vandenberghe.
\newblock \emph{Convex optimization}.
\newblock London, UK: Cambridge university press, 2004.

\bibitem[Bristow et~al.(2006)Bristow, Tharayil, and
  Alleyne]{DABristow2006IEEECSM}
D.~A. Bristow, M.~Tharayil, and A.~G. Alleyne.
\newblock A survey of iterative learning control.
\newblock \emph{IEEE Control Systems Magazine}, 26\penalty0 (3):\penalty0
  96--114, 2006.

\bibitem[Chu and Rapisarda(2023)]{BChu2023CDC}
B.~Chu and P.~Rapisarda.
\newblock Data-driven iterative learning control for continuous-time systems.
\newblock In \emph{Proceedings of 2023 62nd IEEE Conference on Decision and
  Control (CDC)}, pages 4626--4631, 2023.

\bibitem[Dorri et~al.(2018)Dorri, Kanhere, and Jurdak]{ADorri2018IEEEACCESS}
A.~Dorri, S.~S. Kanhere, and R.~Jurdak.
\newblock Multi-agent systems: {A} survey.
\newblock \emph{IEEE Access}, 6:\penalty0 28573--28593, 2018.

\bibitem[He et~al.(2018)He, Wang, Dai, and Luo]{SHe2018TII}
S.~He, M.~Wang, S.~L. Dai, and F.~Luo.
\newblock Leader--follower formation control of {USV}s with prescribed
  performance and collision avoidance.
\newblock \emph{IEEE Transactions on Industrial Informatics}, 15\penalty0
  (1):\penalty0 572--581, 2018.

\bibitem[Hou and Wang(2013)]{ZHou2013InformationScience}
Z.-S. Hou and Z.~Wang.
\newblock From model-based control to data-driven control: {S}urvey,
  classification and perspective.
\newblock \emph{Information Sciences}, 235:\penalty0 3--35, 2013.

\bibitem[Jiang et~al.(2020)Jiang, Bian, and Gao]{ZJiang2020FTSC}
Z.~P. Jiang, T.~Bian, and W.~Gao.
\newblock Learning-based control: {A} tutorial and some recent results.
\newblock \emph{Foundations and Trends{\textregistered} in Systems and
  Control}, 8\penalty0 (3):\penalty0 176--284, 2020.

\bibitem[Landau et~al.(2011)Landau, Lozano, M'Saad, and
  Karimi]{ILandau2011book}
I.~D. Landau, R.~Lozano, M.~M'Saad, and A.~Karimi.
\newblock \emph{Adaptive control: Algorithms, analysis and applications}.
\newblock London, UK: Springer Science \& Business Media, 2011.

\bibitem[Li(2017)]{YLi2017arXiv}
Y.~Li.
\newblock Deep reinforcement learning: {A}n overview.
\newblock \emph{arXiv preprint arXiv:1701.07274}, 2017.

\bibitem[Ortega et~al.(2003)Ortega, Jeltsema, and Scherpen]{ROrtega2003TAC}
R.~Ortega, D.~Jeltsema, and J.~M.~A. Scherpen.
\newblock Power shaping: {A} new paradigm for stabilization of nonlinear {RLC}
  circuits.
\newblock \emph{IEEE Transactions on Automatic Control}, 48\penalty0
  (10):\penalty0 1762--1767, 2003.

\bibitem[Persis and Tesi(2019)]{CDPersis2019TAC}
C.~De Persis and P.~Tesi.
\newblock Formulas for data-driven control: {S}tabilization, optimality, and
  robustness.
\newblock \emph{IEEE Transactions on Automatic Control}, 65\penalty0
  (3):\penalty0 909--924, 2019.

\bibitem[Plesn{\'\i}k(2007)]{JPlesnik2007LAIA}
J.~Plesn{\'\i}k.
\newblock Finding the orthogonal projection of a point onto an affine subspace.
\newblock \emph{Linear Algebra and Its Applications}, 422\penalty0
  (2-3):\penalty0 455--470, 2007.

\bibitem[Poveda et~al.(2019)Poveda, Benosman, and Teel]{JPoveda2019IJACSP}
J.~I. Poveda, M.~Benosman, and A.~R. Teel.
\newblock Hybrid online learning control in networked multiagent systems: A
  survey.
\newblock \emph{International Journal of Adaptive Control and Signal
  Processing}, 33\penalty0 (2):\penalty0 228--261, 2019.

\bibitem[Wang et~al.(2016)Wang, Guo, Liang, Chen, Hu, and Leang]{HWang2016TIE}
H.~Wang, D.~Guo, X.~Liang, W.~Chen, G.~Hu, and K.~K. Leang.
\newblock Adaptive vision-based leader--follower formation control of mobile
  robots.
\newblock \emph{IEEE Transactions on Industrial Electronics}, 64\penalty0
  (4):\penalty0 2893--2902, 2016.

\bibitem[Willems(2007)]{JCWillems2007CSM}
J.~C. Willems.
\newblock The behavioral approach to open and interconnected systems.
\newblock \emph{IEEE Control Systems Magazine}, 27\penalty0 (6):\penalty0
  46--99, 2007.

\bibitem[Zhuang et~al.(2023)Zhuang, Tao, Chen, Stojanovic, and
  Paszke]{ZZhuang2023TSMC}
Z.~Zhuang, H.~Tao, Y.~Chen, V.~Stojanovic, and W.~Paszke.
\newblock An optimal iterative learning control approach for linear systems
  with nonuniform trial lengths under input constraints.
\newblock \emph{IEEE Transactions on Systems, Man, and Cybernetics: Systems},
  53\penalty0 (6):\penalty0 3461--3473, 2023.

\end{thebibliography}

\appendix
\section*{Appendix A: Proof to Theorem \ref{theorem-calculationofsimilarityindexes}}
Since the LAE in (\ref{eq-criterion for similarity}) is solvable, the admissible behaviors $\mathcal{B}_{1,x_{10}}$ and $\mathcal{B}_{2,x_{20}}$ are similar according to Lemma \ref{lemma-similaritycriterion}. Based on this fact, the similarity indexes $\textbf{SI}\left(\mathcal{B}_{1,x_{10}},\mathcal{B}_{2,x_{20}}\right)$ can be further calculated. Following the offline test principles (\ref{eq-principle1}) and (\ref{eq-principle2}) and the matrix $H_i$ in (\ref{eq-unitbasesofW}), the admissible behavior $\mathcal{B}_{i,x_{i0}},\ i\in\{1,2\}$ can be decomposed as
\begin{equation}\nonumber
	\mathcal{B}_{i,x_{i0}}=\mathcal{W}_i+w_{i}^0, \ i\in\{1,2\}
\end{equation}
where
\begin{equation}\nonumber
	\mathcal{W}_i=\text{span}\left(H_i\right).
\end{equation}
According to the definitions of singular values and singular vectors, the $k$-th biggest singular value of the matrix $H_1^{\rm T}H_2$ can be expressed as
\begin{equation} \label{eq-defi2singularvalues}
	s_k=\max_{\left\|l\right\|=\left\|v\right\|=1} l^{\rm T}H_1^{\rm T}H_2v=l_k^{\rm T}H_1^{\rm T}H_2v_k,\ k\in\mathbb{Z}_{n_uT}\backslash\{0\}
\end{equation}
subject to
\begin{equation} \nonumber
	\langle l,l_i\rangle=\langle v,v_i\rangle=0,\ i\in\mathbb{Z}_{k-1}\backslash\{0\}
\end{equation}
where $l_i\in\mathbb{R}^{n_uT}$ and $v_i\in\mathbb{R}^{n_uT}$. From the data-based construction of $H_i$ in (\ref{eq-unitbasesofW}), the matrices $H_1$ and $H_2$ are both orthogonal matrices. By introducing the following coordinate transformation
\begin{equation} \label{eq-pf2theorem1_eq1}
	\begin{aligned}
		p_i&=H_1l_i\in\mathcal{W}_1,\ p=H_1l\in\mathcal{W}_1\\
		q_i&=H_2v_i\in\mathcal{W}_2,\ q=H_2v\in\mathcal{W}_2
	\end{aligned},\ i\in \mathbb{Z}_{k-1}\backslash\{0\}
\end{equation}
it is directly concluded that
\begin{equation} \nonumber
	\begin{aligned}
		\left\|p\right\|&=\left\|H_1l\right\|=\left\|l\right\|=1,\\
		\left\|q\right\|&=\left\|H_2v\right\|=\left\|v\right\|=1
	\end{aligned}
\end{equation}
and
\begin{equation}\nonumber
	\begin{aligned}
		\langle x,x_i\rangle&=\langle H_1l,H_1l_i\rangle=\langle l,l_i\rangle=0\\
		\langle y,y_i\rangle&=\langle H_2v,H_2v_i\rangle=\langle v,v_i\rangle=0
	\end{aligned},\ i\in\mathbb{Z}_{k-1}\backslash\{0\}.
\end{equation}
By substituting (\ref{eq-pf2theorem1_eq1}) into (\ref{eq-defi2singularvalues}), the $k$-th biggest singular value $s_k$ defined in (\ref{eq-defi2singularvalues}) can be alternatively represented by
\begin{equation} \nonumber
	s_k=\max_{p\in\mathcal{W}_1}\max_{q\in\mathcal{W}_2}\langle p,q\rangle=\langle p_k,q_k\rangle,\ k\in\mathbb{Z}_{n_uT}\backslash\{0\}
\end{equation}
subject to
\begin{equation} \nonumber
	\left\|p\right\|=\left\|q\right\|=1,\ \langle p,p_i\rangle=0,\ \langle q,q_i\rangle=0,\ i\in\mathbb{Z}_{k-1}\backslash\{0\}.
\end{equation}
From Definition \ref{defi-principalangles}, it can be observed that the $k$-th biggest singular value of the matrix $H_1^{\rm T}H_2$ is exactly the cosine of $k$-th smallest principal angle between the subspaces $\mathcal{W}_1$ and $\mathcal{W}_2$, that is,
\begin{equation} \nonumber
	s_k=\cos\left(\theta_k\right),\ k\in\mathbb{Z}_{n_uT}\backslash\{0\}.
\end{equation}
Therefore, the similarity indexes between admissible behaviors $\mathcal{B}_{1,x_{10}}$ and $\mathcal{B}_{2,x_{20}}$ (or equivalently, the principal angles between the subspace components $\mathcal{W}_1$ and $\mathcal{W}_2$) can be efficiently obtained through computing the singular values of the matrices $H_1^{\rm T}H_2$, that is,
\begin{equation} \nonumber
	\textbf{SI}\left(\mathcal{B}_{1,x_{10}},\mathcal{B}_{2,x_{20}}\right)=\begin{bmatrix}
		s_1,\ s_2,\ \cdots,\ s_{n_uT}
	\end{bmatrix}.
\end{equation}
Additionally, of note is that the vectors $l_i$ and $v_i$ in (\ref{eq-defi2singularvalues}) are essentially the $i$-th column of the orthogonal matrices $U$ and $V$, that is,
\begin{equation} \nonumber
	\begin{aligned}
		U&=\begin{bmatrix}
			l_1,\ l_2,\ \cdots,\ l_{n_uT}
		\end{bmatrix},\\ V&=\begin{bmatrix}
			v_1,\ v_2,\ \cdots,\ v_{n_uT}
		\end{bmatrix}.
	\end{aligned}
\end{equation}
From Definition \ref{defi-principalangles}, the principal vectors associated with the subspaces $\mathcal{W}_1$ and $\mathcal{W}_2$ can be obtained from
\begin{equation} \nonumber
	\begin{aligned}
		\begin{bmatrix}
			p_1,\ p_2,\ \cdots,\ p_{n_uT}
		\end{bmatrix}&=H_1\begin{bmatrix}
			l_1,\ l_2,\ \cdots,\  l_{n_uT}
		\end{bmatrix}=H_1U,\\
		\begin{bmatrix}
			q_1,\ q_2,\ \cdots,\ q_{n_uT}
		\end{bmatrix}&=H_2\begin{bmatrix}
			v_1,\ v_2,\ \cdots,\ v_{n_uT}
		\end{bmatrix}=H_2V.
	\end{aligned}
\end{equation}
Consequently, the similarity indexes between two admissible behaviors and the principal vectors between their associated subspace components can be efficiently calculated through the SVD of $H_1^{\rm T}H_2$. The proof is completed.

\section*{Appendix B: Proof to Theorem \ref{theorem-similaritybasedlearning}}
Following the conditions C1) and C2), by leveraging the sampled data $\left(U_i^{Test},Y_i^{Test}\right)$, the admissible behaviors $\mathcal{B}_{1,x_{10}}$ and $\mathcal{B}_{2,x_{20}}$ can be decomposed as
\begin{equation} \nonumber
	\mathcal{B}_{i,x_{i0}}=\text{span}\left(H_i\right)+w_i^0,\ \forall i\in\{1,2\}
\end{equation}
where $H_i\in\mathbb{R}^{n_wT\times n_uT}$ satisfies $H^{\rm T}_iH_i=I_{n_uT}$.
As previously emphasized, the to-be-sought $w_h$ in Problem 3) is essentially the orthogonal projection of $w_g$ onto $\mathcal{B}_{1,x_{10}}$, i.e., $P_{\mathcal{B}_{1,x_{10}}}\left(w_g\right)$. From the conditions C3), the admissible behaviors $\mathcal{B}_{1,x_{10}}$ and $\mathcal{B}_{2,x_{20}}$ are similar, which allows for further calculating the similarity indexes. Additionally, by leveraging Theorem \ref{theorem-calculationofsimilarityindexes}, the condition C4) ensures that the similarity indexes between $\mathcal{B}_{1,x_{10}}$ and $\mathcal{B}_{2,x_{20}}$ can be obtained via calculating the singular values of $H^{\rm T}_1H_2$ as
\begin{equation}\nonumber
	\textbf{SI}\left(\mathcal{B}_{1,x_{10}},\mathcal{B}_{2,x_{20}}\right)=\begin{bmatrix}
		s_1,\ s_2,\ \cdots,\ s_{n_uT}
	\end{bmatrix}.
\end{equation}
Another helpful conclusion brought by the conditions C4) is that the principal vectors associated with the subspaces $\mathcal{W}_1$ and $\mathcal{W}_2$ can be calculated as $H_1U$ and $H_2V$. Consequently, the admissible behaviors can be equivalently expressed as
\begin{equation}\nonumber
	\begin{aligned}
		\mathcal{B}_{1,x_{10}}&=\text{span}\left(H_1U\right)+w_1^0,\\
		\mathcal{B}_{2,x_{20}}&=\text{span}\left(H_2V\right)+w_2^0.
	\end{aligned}
\end{equation}
From Lemma \ref{lemma-data-based representation},
for any admissible trajectory $w_g\in\mathcal{B}_{2,x_{20}}$, there always exists some vector $\overline{g}\in\mathbb{R}^{n_uT}$ such that (\ref{eq-representation of w_g}) holds. Building upon this fact, the existence of $\overline{g}$ can be guaranteed, and the to-be-sought $w_h\in\mathcal{B}_{1,x_{10}}$ can be rewritten as
\begin{equation} \nonumber
	\begin{aligned}
		w_h&=P_{\mathcal{B}_{1,x_{10}}}\left(w_g\right) \\
		&=P_{\mathcal{B}_{1,x_{10}}}\left(H_2V\overline{g}+w_2^0\right).
	\end{aligned}
\end{equation}
Nevertheless, the operator $P_{\mathcal{B}_{1,x_{10}}}\left(\cdot\right)$ is not linear, which results in difficulties in calculating this orthogonal projection. From Lemma \ref{lemma-projection}, the orthogonal projection onto $\mathcal{B}_{1,x_{10}}$ can be equivalently expressed as
\begin{equation} \nonumber
	\begin{aligned}
		P_{\mathcal{B}_{1,x_{10}}}\left(H_2V\overline{g}+w_2^0\right)&=w_1^0+P_{\text{span}\left(H_1U\right)}\left(H_2V\overline{g}+w_2^0-w_1^0\right).
	\end{aligned}
\end{equation}
Thanks to the linearity of the operator $P_{\text{span}\left(H_1U\right)}\left(\cdot\right)$, $w_h$ can be further rewritten as
\begin{equation} \nonumber
	w_h=w_1^0+P_{\text{span}\left(H_1U\right)}\left(H_2V\right)\overline{g}+P_{\text{span}\left(H_1U\right)}\left(w_2^0-w_1^0\right).
\end{equation}
Afterward, we focus on calculating $P_{\text{span}\left(H_1U\right)}\left(H_2V\right)\overline{g}$.
By leveraging the principal vectors $H_1U$ and $H_2V$ and similarity indexes $\textbf{SI}\left(\mathcal{B}_{1,x_{10}},\mathcal{B}_{2,x_{20}}\right)$, the orthogonal projection $P_{\text{span}\left(H_1U\right)}\left(H_2V\right)\overline{g}$ can actually be efficiently computed. This is also the prominent advantage of the similarity-based learning compared to existing learning-based methods. Let the $i$-th column of $H_1U$ (or $H_2V$) be denoted as $\left(H_1U\right)_i$ (or $\left(H_2V\right)_i$), then $P_{\text{span}\left(H_1U\right)}\left(H_2V\right)$ can be expressed as
\begin{equation} \nonumber
	\begin{aligned}
		&P_{\text{span}\left(H_1U\right)}\left(H_2V\right)\\
		=&\begin{bmatrix}
			P_{\text{span}\left(H_1U\right)}\left(H_2V\right)_1,\ \cdots,\
			P_{\text{span}\left(H_1U\right)}\left(H_2V\right)_{n_uT}
		\end{bmatrix}
	\end{aligned}
\end{equation}
where
\begin{equation} \nonumber
	P_{\text{span}\left(H_1U\right)}\left(H_2V\right)_i=\sum_{j=1}^{n_uT}\langle\left(H_2V\right)_i,\left(H_1U\right)_j\rangle\left(H_1U\right)_i
\end{equation}
holds for all $i\in\mathbb{Z}_{n_uT}\backslash\{0\}$. From the Definitions \ref{defi-principalangles} and \ref{defi-similarityindexes} and the properties of SVD, it follows that
\begin{equation} \label{eq-principal vectors}
	\begin{aligned}
		s_k&=\langle \left(H_2V\right)_k,\left(H_1U\right)_k\rangle, \ \forall k\in\mathbb{Z}_{n_uT}\backslash\{0\},\\
		0&=\langle \left(H_2V\right)_i,\left(H_1U\right)_j\rangle,\ \forall i\neq j,\ \forall i,j\in\mathbb{Z}_{n_uT}\backslash\{0\}.
	\end{aligned}
\end{equation}
By leveraging (\ref{eq-principal vectors}), the orthogonal projection $P_{\text{span}\left(H_1U\right)}\left(H_2V\right)$ can be further expressed as
\begin{equation} \nonumber
	\begin{aligned}
		&P_{\text{span}\left(H_1U\right)}\left(H_2V\right)\\
		=&\begin{bmatrix}
			P_{\text{span}\left(H_1U\right)}\left(H_2V\right)_1,\ \cdots,\
			P_{\text{span}\left(H_1U\right)}\left(H_2V\right)_{n_uT}
		\end{bmatrix}\\
		=&\begin{bmatrix}
			\left(H_1U\right)_1s_1,\ \cdots,\ \left(H_1U\right)_{n_uT}s_{n_uT}
		\end{bmatrix}\\
		=&H_1UD.
	\end{aligned}
\end{equation}
Therefore, the to-be-sought admissible trajectory $w_h\in\mathcal{B}_{1,x_{10}}$ can be further expressed as
\begin{equation} \nonumber
	w_h=w_1^0+H_1UD\overline{g}+P_{\text{span}\left(H_1\right)}\left(w_2^0-w_1^0\right).
\end{equation}
Since $w_h$ is essentially the orthogonal projection of $w_g$ onto $\mathcal{B}_{1,x_{10}}$, the difference $\left\|w_h-w_g\right\|$ must be minimal. The proof is completed.
\end{document}